\documentclass[a4paper,USenglish]{lipics-v2021}
\hideLIPIcs
\nolinenumbers
\usepackage[utf8]{inputenc}

\usepackage{url}
\usepackage{cite}
\usepackage{amsmath}
\usepackage{mathrsfs}

\def\apxmode{inline}
\usepackage[appendix=\apxmode,bibliography=common]{apxproof}
\usepackage[\apxmode]{optional}
\newtheoremrep{lem}{Lemma}[section]
\newtheoremrep{thm}[lem]{Theorem}

\usepackage{mathtools}  
\mathtoolsset{showonlyrefs}

\usepackage{tikz-cd}
\tikzcdset{arrow style=tikz, diagrams={>=stealth}}

\newcommand{\exec}{{\mathsf{Exec}}}
\newcommand{\view}{{\mathsf{View}}}
\newcommand{\cview}{{\mathsf{CView}}}
\newcommand{\correct}{{\mathsf{Correct}}}
\newcommand{\participating}{{\mathsf{Part}}}
\newcommand{\compatible}{{\mathsf{Compatible}}}
\newcommand{\decision}{{\mathsf{Decision}}}
\newcommand{\timecview}[2]{{\mathsf{CView}}_{#1}|_{#2}}

\newcommand{\m}[1]{\ensuremath{\mathcal{#1}}}

\newcommand{\init}{{\mathsf{Init}}}

\newcommand{\inputs}{{\mathcal{I}}}
\newcommand{\outputs}{{\mathcal{O}}}
\newcommand{\task}{{T}}
\newcommand{\protocol}{{\mathcal{P}}}
\newcommand{\timeprotocol}[1]{{\mathcal{P}}|_{#1}}
\newcommand{\terminating}{{\mathscr{T}}}
\newcommand{\Vin}{{V^{\mathrm{in}}}}
\newcommand{\Vout}{{V^{\mathrm{out}}}}

\DeclareMathOperator{\Chr}{Chr}
\DeclareMathOperator{\sta}{st}

\DeclareMathOperator{\diam}{diam}
\DeclareMathOperator{\id}{id}

\newcommand{\figref}[1]{{Fig.~\ref{#1}}}

\newcommand{\ie}{{\emph{i.e.}}}
\newcommand{\eg}{{\emph{e.g.}}}

\title{Topological Characterization of Task Solvability in General Models of Computation}
\titlerunning{Task Solvability in General Models of Computation}
\author{Hagit Attiya}{Department of Computer Science, Technion, Israel}{hagit@cs.technion.ac.il}{https://orcid.org/0000-0002-8017-6457}{partially supported by the Israel Science Foundation (grants 380/18 and 22/1425)}
\author{Armando Castañeda}{Instituto de Matem\'aticas, Universidad Nacional Aut\'onoma de M\'exico, Mexico}{armando.castaneda@im.unam.mx}{https://orcid.org/0000-0002-8017-8639}{partially supported by the DGAPA PAPIIT project IN108723}
\author{Thomas Nowak}{Laboratoire Méthodes Formelles, Université Paris-Saclay, CNRS, ENS Paris-Saclay, France \and Institut Universitaire de France, France}{thomas@thomasnowak.net}{https://orcid.org/0000-0003-1690-9342}{partially supported by the ANR grant ANR-21-CE48-0003}
\authorrunning{Attiya, Castañeda and Nowak}
\Copyright{Hagit Attiya, Armando Castañeda and Thomas Nowak}
\ccsdesc{Theory of computation~Distributed computing models}
\keywords{task solvability, combinatorial topology, point-set topology}

\begin{document}

\maketitle

\begin{abstract}
The famous \emph{asynchronous computability theorem} (ACT) relates the existence 
of an asynchronous wait-free shared memory protocol 
for solving a task with the existence of a simplicial map
from a subdivision of the simplicial complex representing the inputs to 
the simplicial complex
representing the allowable outputs.
The original theorem relies on a correspondence between protocols 
and simplicial maps in round-structured models of computation that induce 
a compact topology.
This correspondence, however, is far from obvious for computation models that induce a 
non-compact topology, and indeed previous attempts to extend the ACT have failed. 

This paper shows that in \emph{every} non-compact model,
protocols solving tasks correspond to simplicial maps that need to be \emph{continuous}.
It first proves a \emph{generalized} ACT for sub-IIS models,
some of which are non-compact,
and applies it to the \emph{set agreement} task. 
Then it proves that in general models too, protocols are simplicial
maps that need to be continuous, hence showing that 
the topological approach is \emph{universal}.
Finally, it shows that 
the approach used in ACT that equates protocols and simplicial complexes
actually works for \emph{every} compact model.

Our study combines, for the first time, combinatorial and point-set topological aspects of 
the executions admitted by the computation model. 
\end{abstract}

\newpage

\section{Introduction}

The celebrated topological approach in distributed computing relates 
task solvability to the topology of inputs and outputs of the task
and the topology of the protocols allowed in a particular
model of computation. 
This approach rests on three pillars.
First, \emph{configurations}, 
whether of inputs, outputs or protocol states, 
can be modeled as \emph{simplexes}, which are finite sets.
Second, the inherent \emph{indistinguishability} of configurations
is crisply captured by intersections between simplexes.
Third, a \emph{carrier map} captures the notion of the set of 
configurations that are reachable from a given configuration.

More concretely, in this approach,
\emph{tasks} are triples $T = (\inputs, \outputs, \Delta)$,
where $\inputs$ and $\outputs$ are \emph{simplicial complexes} modeling 
the inputs and outputs of the task, and 
$\Delta$ is a \emph{carrier map} specifying the possible valid outputs, $\Delta(\sigma)$, 
for each input simplex $\sigma \in \inputs$.
Similarly, \emph{protocols} are triples $P = (\inputs, \protocol, \Xi)$,
where $\protocol$ is the complex modeling the final configurations of the protocol,
and $\Xi$ is a carrier map specifying the reachable final configurations, $\Xi(\sigma)$,
from $\sigma$.

With this perspective in mind, it is natural to conclude that 
a protocol maps final states (\ie, states in final configurations) to outputs, 
and for the protocol to be correct, the mapping must be \emph{simplicial};
that is, all outputs of final states in the same simplex $\tau \in \Xi(\sigma)$ 
(\ie, in the same final configuration) must be in the same output simplex of $\Delta(\sigma)$.
Thus, a protocol induces a simplicial map from $\protocol$ to $\outputs$. 
Moreover, since decisions 
of processes are only based on local information, it is natural
to conclude the converse, \ie, \emph{any} simplicial map implies
a protocol.
(In general, a protocol specifies also the communication during an execution.
However, solvability can consider only existence of a decision function,
by assuming that the protocol is \emph{full-information}.)
This leads to the following
purely topological solvability characterization:
a protocol $P = (\inputs, \protocol, \Xi)$ solves 
a task $T = (\inputs, \outputs, \Delta)$ \emph{if and only if} 
there is a simplicial map $\delta: \protocol \rightarrow \outputs$
such that for every $\sigma \in \inputs$, we have
$(\delta\circ\Xi)(\sigma) \subseteq \Delta(\sigma)$.

The discussion so far did not depend on a particular model of computation.
Indeed, this approach seems universal and gives the impression 
that protocols and simplicial maps are the same,
and that for all models, the solvability question can be reduced to 
the existence of a simplicial map. 
In fact, the correspondence between protocols and simplical maps 
seems so self-evident that frequently the characterization above seems to
require no proof, and is introduced as a definition 
(\eg,~\cite[Section~4.2.2; Definition~8.4.2]{HKR-book}).

This approach works well for cases where the model of computation
has a particular round structure\footnote{
Rounds may be explicit, 
like in the synchronous message-passing $t$-resilient model, 
or implicit, for example, modeled by \emph{layers}, 
as in the Iterated Immediate Snapshot model.}
and it induces a \emph{compact} topology so that the correspondence 
between protocols and simplicial maps holds.
Roughly speaking, a compact topological space 
has no ``punctures'' or ``missing endpoints'', namely,
it does not exclude any limit point.
If a model of computation is specified 
as a set of infinite executions, 
then a compact model will contain all its ``limit executions''.
For example, the \emph{Iterated Immediate Snapshot} (IIS) model ensures that 
computation proceeds in sequence of (implicit) rounds;
in each round, any of a \emph{finite} set of possible schedules can happen.
Thus, the model contains every infinite execution with this round structure.
Models like IIS are sometimes called \emph{oblivious}~\cite{CGP15},
and are known to induce \emph{finite complexes with compact topology,
where protocols and simplicial maps are the same}.
Round-structured compact models 
have been extensively studied in the literature,
and different techniques have been developed for them (\eg,~\cite{AR02,CFPRRT2021,FGL21,HR12}). 

However, this approach is not true in all models,
specifically, in non-compact ones.
In a non-compact model,
typically some ``good'' schedules only \emph{eventually} happen,
which then implies that the model is not limit-closed.
Examples of a non-compact models are $t$-resilient \emph{asynchronous} models
where any process is guaranteed to eventually obtain information 
from at least $t-1$ other processes infinitely often, 
but the process can take an unbounded number
of steps before that happens. 
This means, for example, that the model contains every
infinite execution where a process runs solo for a \emph{finite} number of
steps and then obtains information from $t-1$ other processes,
but it does not contain the infinite solo execution of the process, 
\ie, the limit execution.

Challenging non-compact models have been mostly treated in the literature indirectly,
through ``compactification''.
Sometimes, compactification consists of considering only
protocols with a concrete round structure,
as is done in some chapters of~\cite{HKR-book}.
In other cases, a computationally-equivalent round-structured compact model is 
analyzed instead (\eg,~\cite{KRH18, SHG16}).
This requires first to prove that the models are equivalent, 
through simulations in both directions, 
and then characterize solvability in the compact model.
For impossibility results, it is also sometimes possible to identify a compact sub-model in which a problem of interest is still unsolvable~\cite{LM95}.

Round-structured compact models have also served to analyze 
other compact models. 
In the famous \emph{asynchronous computability theorem} 
(ACT)~\cite{HS1999},
fully characterizing solvability of the non-round-structured 
compact read/write wait-free shared memory model, 
a crucial step is showing equivalence with the IIS model.
This restricts the solvability characterization to
subdivisions, which are well-behaved topological spaces,
making the ACT highly useful.

Some sub-IIS models, 
with subsets of IIS like $t$-resilient computations, 
are non-compact. 
For this reason, 
an attempt~\cite{GKM14} to generalize the ACT to arbitrary sub-IIS models and tasks
had to directly address non-compact models.
The idea of this so-called \emph{generalized} ACT is to 
somehow reuse the nice structure of IIS, 
modeling any sub-IIS model as a possibly infinite subdivision.

This raises two important questions that have not been explicitly investigated so far,
which are addressed in this paper. 
(1) Can protocols in all models of distributed computation
be captured as simplicial maps? 
(2) Can the topological approach be applied to all models of computation? 
The answers to these questions are not self-evident since there may be non-compact
models or non-round-structured compact models, which cannot be compactified.

We note that there is already a recent negative answer to the first question:
Godard and Perdereau~\cite{GP20} showed a non-compact model where consensus 
is unsolvable, but nevertheless,
it has a simplicial map as described above.
Roughly, it considers a sub-IIS model where a single infinite execution
of IIS is removed.
The resulting model is not compact. It turns out that the complex of a protocol is an \emph{infinite} subdivision that is \emph{disconnected},
hence, there is a simplicial map from it to the consensus output complex (which is disconnected too).
But the map does not imply a consensus protocol, 
since intuitively, the decisions must be consistent as they approach 
the discontinuity of the removed execution.
More specifically, the simplicial map does not imply a protocol because it is \emph{not continuous}. 
Section~\ref{sec:need-conitnuity} details the example based on their ideas.
While continuity of simplicial maps is guaranteed in compact models,
this is not the case in non-compact ones.
This example demonstrates that the generalization in~\cite{GKM14} 
is flawed as it misses the continuity property of simplicial maps.
Godard and Perdereau also correct this problem 
\emph{for the special case of two processes and the consensus task}.

Continuity of simplicial maps may seem trivial, 
but it was overlooked for long time, before~\cite{GP20}.
Here, we further expose its importance.

We first study task solvability in the well-structured 
simplicial complexes induced by sub-IIS models. 
Our first contribution (Theorem~\ref{thm:GACT-fixed}) is to 
present a correct generalized ACT for any number of processes 
and arbitrary tasks.
Our approach 
is motivated by the critical role of continuity.
Our second contribution (Theorem~\ref{theo-set-agreement}) is to use
our generalized ACT theorem in order to provide an 
impossibility condition for set agreement in sub-IIS models,
where the continuity requirement of simplicial maps
allows a natural generalization of the known
impossibility conditions for round-structured compact models.

While this settles the questions for sub-IIS models, 
the questions for general models remain open.
Our third contribution (Theorem~\ref{thm:topology-approach}) 
shows that the topological approach is applicable  
in \emph{all} models of computation, if one requires simplicial maps
to be continuous.
Unlike the case of round-structured compact and sub-IIS models,
proving the applicability of the topology approach to general
non-compact models is not straightforward.
It requires to combine \emph{point-set topological} techniques~\cite{AS84}
with \emph{combinatorial topology} techniques~\cite{HKR-book}.

We use this result in our fourth contribution: 
a proof that the approach described at the 
beginning of the introduction, 
equating protocols and simplicial maps that are not required to be 
continuous, is universal for compact models 
(Theorem~\ref{thm:projective:limit:compact}). 
Namely, in every compact model, possibly non-round-structured,
it is indeed the case that there is correspondence between
protocols and simplicial maps, hence the approach works in all 
these cases.
The proof of this result is far from trivial, 
and it uses \emph{projective limits} 
from \emph{category theory}~\cite{MacLane78}.

As far as we know, non-compact models have been directly studied
only in~\cite{NSW19:podc,GP20,FevatG2011minimal,GKM14,CoutoulyG2023}.
A full combinatorial solvability characterization for two-process consensus 
under synchronous general message-loss failures appears 
in~\cite{FevatG2011minimal}. 
For the case of two processes, these models are all sub-IIS, 
hence this work is the first that directly studies non-compact models. 
Then, \cite{GKM14} attempted to generalize ACT to
general sub-IIS models and tasks, for any number of processes.
The solvability of two-process consensus is studied again in~\cite{GP20}, 
now from a combinatorial topology perspective, 
where it is shown that the attempt in~\cite{GKM14} is flawed.
That paper also provides an alternative full topological solvability 
characterization for two-process consensus.
Recently, sub-IIS models were studied through \emph{geometrization}~\cite{CoutoulyG2023}, 
\ie, using a mapping from IIS executions to points in the Euclidean space, 
which in turn induces a topology. 
The geometrization is used to derive a full solvability characterization 
for \emph{set agreement} in sub-IIS models, and it generalizes the two-process
consensus solvability characterization of~\cite{GP20}.
A solvability characterization for consensus (only) in general models, 
for any number of processes, is presented in~\cite{NSW19:podc}.
It is derived using point-set topology techniques from~\cite{AS84}, 
without combining them with combinatorial topology.
Recent formalizations~\cite{AAEGZ19, ACR23} for proofs based on valency 
arguments show that for some tasks, \eg, set agreement and renaming,
impossibility cannot be shown by inductively constructing infinite executions. 
This means that arguments regarding the final protocol states
are necessary in order to prove impossibility.  
Our results indicate that such proofs can be carried within 
combinatorial topology, in general models of computation.

In summary, our contributions are:
\begin{enumerate}
    \item A generalized ACT for arbitrary sub-IIS models 
    (Theorem~\ref{thm:GACT-fixed}).

    \item An application of the generalized ACT to set agreement (Theorem~\ref{theo-set-agreement}).

    \item A proof that if simplicial maps from $\protocol$ to $\outputs$ are required to be continuous, 
    the topological approach works 
    for every model of computation (Theorem~\ref{thm:topology-approach}).

    \item A proof that the usual topological approach where simplicial maps
    are not required to be continuous works 
    for every compact model (Theorem~\ref{thm:projective:limit:compact}).

\end{enumerate}

\section{Preliminaries}

This section presents the elements of combinatorial topology and point set topology
used in further sections,
and defines tasks, system models and task solvability.

We start by fixing some basic notation.
We denote by~$\Pi$ the set of processes and let $n = |\Pi|$.
For any function $f:X\to Y$ and subsets $A\subseteq X$ and $B\subseteq Y$, 
we denote by~$f[A]$ the image of the set~$A$ under~$f$ and by~$f^{-1}[B]$ 
the inverse image of the set~$B$ under~$f$.

\subsection{Elements of Combinatorial Topology and Decision Tasks}

To be the most general possible, we use the language of 
colored tasks~\cite[Definition~8.2.1]{HKR-book}, to study one-shot distributed
decision tasks like consensus or set agreement.
We use the standard concepts in \cite{HKR-book} with the only difference that 
simplicial complexes might be infinite, 
\ie, a possibly infinite sets of finite sets.

\begin{toappendix}
A \emph{simplicial} complex 
is a (possibly infinite) set $V$ along with a  (possibly infinite) collection $\m{K}$  
of finite subsets of $V$ closed under containment
\ie, if $\sigma \in \m{K}$ then $\sigma' \in \m{K}$, for any $\sigma' \subseteq \sigma$.
An element of $V$ is called a \emph{vertex} of $\m{K}$, 
and the vertex set of $\m{K}$ is denoted by $V(\m{K})$.
Each set in $\m{K}$ is called a \emph{simplex}.
A subset of a simplex is called a \emph{face} of that simplex.
The \emph{dimension of a simplex} $\sigma$, denoted $\dim \sigma$, is
one less than the number of elements of $\sigma$, \ie, $|\sigma|-1$.
The \emph{dimension of a complex} is the smallest integer that upper bounds the
dimension of any of its simplexes, or $\infty$ if there is no such bound.
A simplex $\sigma$ in $\m{K}$ is called a \emph{facet} of $\m{K}$ 
if $\sigma$ is not properly contained in any other simplex.
A complex is \emph{pure} if all its facets have the same dimension.
We will focus on pure complexes, either finite or infinite.

Let $\m{K}$ be a complex and $\sigma$ be a simplex of it.
The \emph{star} of $\sigma$ in $\m{K}$ is the complex
$\sta \sigma = \{\tau \in \m{K} \mid \sigma\subseteq\tau \}$.

Let $\m{K}$ and $\m{L}$ be complexes.
A \emph{vertex map} from $\m{K}$ to $\m{L}$
is a function $h:V(\m{K})\to V(\m{L})$.
If~$h$ also carries simplexes of $\m{K}$ to simplexes of $\m{L}$,
it is called a \emph{simplicial map}.

For two complexes $\m{K}$ and $\m{L}$, if $\m{K} \subseteq \m{L}$, 
we say $\m{K}$ is a \emph{subcomplex} of $\m{L}$.
Given two complexes $\m{K}$ and $\m{L}$, a \emph{carrier map} 
$\Phi: \m{K} \to 2^{\m{L}}$
maps each simplex $\sigma \in \m{K}$ to a subcomplex $\Phi(\sigma)$ of $\m{L}$, 
such that for every two simplexes $\tau$ and $\tau'$ in $\m{K}$ that satisfy $\tau \subseteq \tau'$, we have $\Phi(\tau) \subseteq \Phi(\tau')$. 
We say that $\Phi$ is \emph{rigid} if
for every $\sigma \in \m{K}$,
$\Phi(\sigma)$ is pure of dimension $\dim \sigma$.

A \emph{geometric realization} of a complex \m{K} is an embedding of the simplexes
of \m{K} into a real vector space such that, roughly speaking, 
intersections of simplexes are respected.
All geometric realizations
of a complex are topologically equivalent, \ie, homeomorphic.
Thus, we speak of \emph{the} geometric realization
of \m{K}, which is denoted $|\m{K}|$.
The standard construction sets $\lvert \m{K}\rvert$ equal to the set of functions $\alpha:V(\m{K})\to[0,1]$ such that $\{ v \in V(\m{K}) \mid \alpha(v) > 0 \}$ is a simplex of~$\m{K}$ and $\lVert \alpha \rVert_1 = \sum_{v\in V(\m{K})} \alpha(v) = 1$.
The $1$-norm induces a metric on~$\lvert \m{K}\rvert$ that makes its diameter equal to~$1$ if~$\m{K}$ has more than one vertex.
Any simplicial map $h:\m{K}\to \m{L}$ induces a function
$|h|:|\m{K}|\to |\m{L}|$. If the complexes are finite, 
then $|h|$ is necessarily continuous, and there is no guarantee of that otherwise.

A \emph{coloring} of a complex $\m{K}$ is a function $\chi: V(\m{K}) \rightarrow \Pi$.
The coloring is \emph{chromatic} if any two distinct vertices of the
same facet of $\m{K}$ have distinct colors.
A \emph{chromatic} complex is a simplicial complex
equipped with a chromatic coloring.
A \emph{labeling} of a complex $\m{K}$ is a function
$\ell: V(\m{K}) \rightarrow L$, where $L$ is a set.
The set $L$ will be a set of inputs, outputs or process states.
Below, we will consider chromatic and labeled complexes such that
each vertex is uniquely identified by its color together with its label,
namely, for any two distinct vertices $u$ and $v$, 
$(\chi(u),\ell(u)) \neq (\chi(v),\ell(v))$.
For any vertex $v$ of any such complex, 
we let denote by $v(p,x)$ the unique vertex of the complex 
with color~$p \in \Pi$ and label~$x \in L$.

Let \m{K} be a chromatic complex. 
The \emph{standard chromatic subdivision} of \m{K}, denoted $\Chr \m{K}$,
is the chromatic complex whose vertices have the form $(p, \sigma)$,
where $p \in \Pi$, $\sigma$ is a face of a facet of \m{K}
and $p \in \chi(\sigma)$. A set 
$\{ (p_0, \sigma_0), (p_1, \sigma_1), \hdots, (p_s, \sigma_s) \}$
is a simplex of $\Chr \m{K}$ if and only if 
$\sigma_0 \subseteq \sigma_1 \subseteq \hdots \subseteq \sigma_s$
and for all $0 \leq q,r \leq s$, 
if $q \in \chi(\sigma_r)$ then $\sigma_q \subseteq \sigma_r$.
The chromatic coloring $\chi$ for $\Chr \m{K}$ is defined as 
$\chi(p,\sigma) = p$. 
Figure~\ref{fig:scs} contains the standard chromatic subdivision
of an edge, a 1-dimensional simplex, and a triangle, a 2-dimensional simplex.
The $k$-th standard chromatic subdivision,
$\Chr^k \m{K}$,
is obtained by iterating $k$ times the standard chromatic subdivision.
The standard chromatic subdivision is indeed a subdivision: 
$|\Chr^k \m{K}| \cong |\m{K}|$, for every $k \geq 0$.

\begin{figure}[tb]
\begin{center}
\includegraphics[scale=0.55]{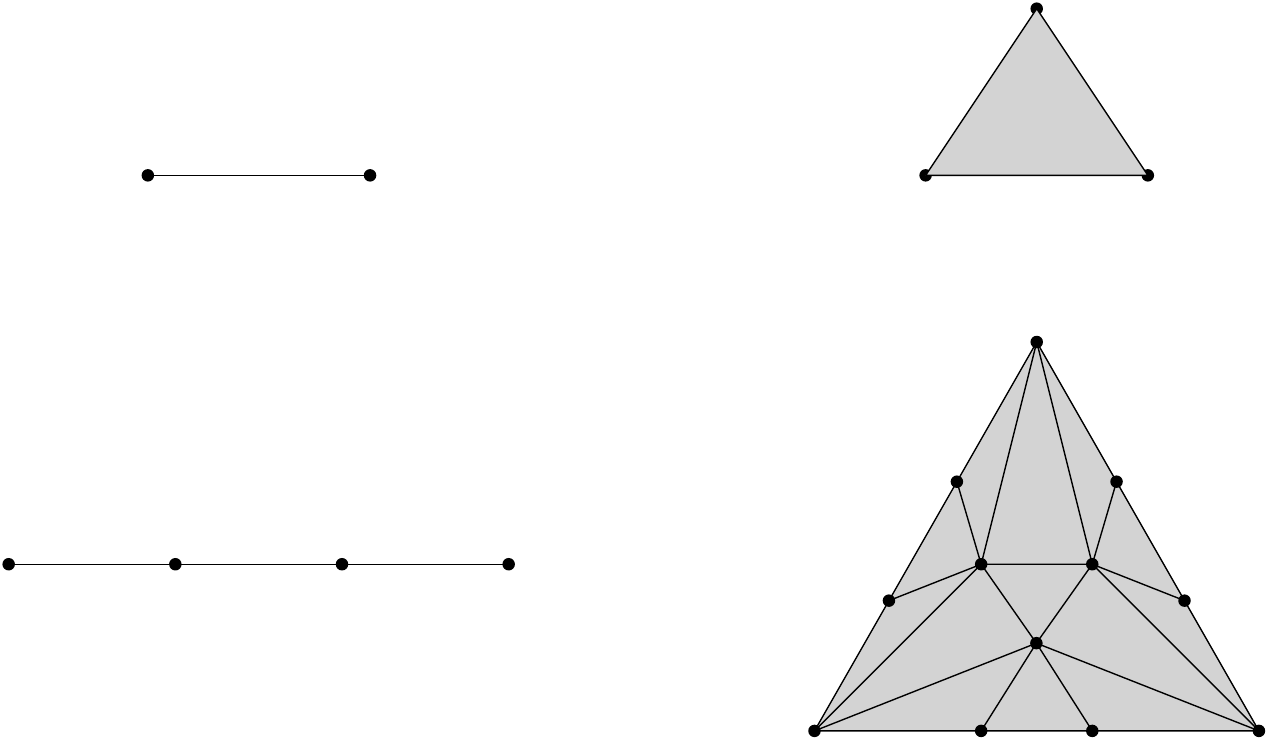}
\caption{The standard chromatic subdivision of an edge and of a triangle.}
\label{fig:scs}
\end{center}
\end{figure}

A simplicial map $h:V(\m{K})\to V(\m{L})$ is \emph{chromatic} 
if it carries colors, \ie, $\chi(v) = \chi(h(v))$,
for every vertex~$v$ of $\m{K}$.
A carrier map $\Phi: \m{K} \to 2^{\m{L}}$ is \emph{chromatic}
if $\Phi(\sigma)$ is pure and chromatic of dimension
$\dim \sigma$, and each facet of it has colors $\chi(\sigma)$.

\end{toappendix}

A \emph{decision task} is a triple $\task = (\inputs, \outputs, \Delta)$ such that:
\begin{itemize}
\item $\inputs$, the \emph{input complex}, is a \emph{finite} pure chromatic simplical complex of dimension $n-1$, 
whose vertices are additionally labeled by a set of inputs~$\Vin$. 
Each simplex of $\inputs$ specifies private inputs for the processes that appear in the simplex.
\item $\outputs$, the \emph{output complex}, is a \emph{finite} pure chromatic simplical complex of dimension $n-1$, 
whose vertices are additionally labeled by a set of inputs~$\Vout$.
As above, each simplex of $\outputs$ specifies private outputs for the processes in the simplex. 
\item $\Delta$ is a chromatic carrier map from~$\inputs$ to~$\outputs$, $\Delta(\sigma)$,
that \emph{specifies} the valid outputs for every input simplex $\sigma$ in $\inputs$.
Namely, when the inputs are the ones specified in $\sigma$, the outputs
in any simplex of $\Delta(\sigma)$ are allowed.
\end{itemize}

\subsection{Elements of Point-Set Topology}

In addition to combinatorial topology, 
we employ point-set topology~\cite{bourbaki:general:topology:1:4}, 
\ie, the general mathematical theory of closeness, convergence, and continuity.
The topologies that we define here are described by \emph{metrics}, 
which are distance functions $d:X\times X\to [0,\infty)$ that satisfy:
\begin{enumerate}
\item Positive definiteness: $d(x,y)=0$ if and only if $x=y$ 
\item Symmetry: $d(x,y) = d(y,x)$
\item Triangle inequality: $d(x,z) \leq d(x,y) + d(y,z)$
\end{enumerate}
A set equipped with a metric is called a \emph{metric space}.
The most basic metric is the \emph{discrete metric},
which is defined by:
\begin{equation}
d(x,y) =
\begin{cases}
0 & \text{if } x=y\\
1 & \text{if } x\neq y
\end{cases}
\end{equation}
That is, the discrete metric can only give the information whether two elements are equal, but implies no finer-grained notion of closeness.

A central notion in point-set topology are \emph{open sets}, which are subsets $O\subseteq X$ such that
\begin{equation}
\forall x\in O \quad \exists \varepsilon > 0 \colon \quad B_\varepsilon(x) \subseteq O
\end{equation}
where
$B_\varepsilon(x) = \{ y\in X \mid d(x,y) < \varepsilon \}$
is the open ball with radius~$\varepsilon$ around~$x$.
With respect to the discrete metric, every subset $O\subseteq X$ is open.
This follows from the fact that the open ball with radius~$1/2$ around~$x$ is equal to $B_{1/2}(x)=\{x\}$, \ie, only contains~$x$ itself.

The general definition of a \emph{topological space} is a nonempty set~$X$ together with a \emph{topology}, \ie, a set $\m{O}\subseteq 2^X$ of subsets of~$X$ that is closed under arbitrary unions and finite intersections.
The elements of~$\m{O}$ are called the open sets of the space.
With the above definition, every metric induces a topology.

A particular class of metrics that we use in this paper is that of \emph{ultrametrics}.
They satisfy the stronger ultrametric triangle inequality: $d(x,z) \leq \max\{ d(x,y) , d(y,z) \}$ for all $x,y,z\in X$.
The discrete metric is an example of an ultrametric.
In an ultrametric space, two open balls are either disjoint or one is a subset of the other, as is shown by the following folklore lemma:

\begin{lemrep}\label{lem:ultrametric:balls}
Let~$X$ be an ultrametric space.
For all $x,y\in X$ and all $\delta,\varepsilon > 0$, one of the following is true: 
(1) $B_\delta(x) \cap B_\varepsilon(y) = \emptyset$,
(2) $B_\delta(x) \subseteq B_\varepsilon(y)$,
(3) $B_\varepsilon(y) \subseteq B_\delta(x)$.
\end{lemrep}
\begin{proof}
Assume that both~(1) and~(2) are false.
We will prove that then~(3) is true.

Let $v\in B_\varepsilon(y)$.
We need to show that $v\in B_\delta(x)$.
Since~(1) is false, there exists a $z\in B_\delta(x) \cap B_\varepsilon(y)$.
Applying the ultrametric triangle inequality twice, we have:
\begin{equation}
\begin{split}
d(v,x)
& \leq
\max\{ d(v,z) , d(z,x) \}
\leq
\max\{ d(v,y) , d(y,z) , d(z,x) \}
\\ & <
\max\{ \varepsilon , \varepsilon , \delta \}
=
\max\{ \varepsilon , \delta \}
\end{split}
\end{equation}
It remains to prove that $\varepsilon\leq \delta$ so that $\max\{\varepsilon,\delta\} = \delta$ and $v\in B_\delta(x)$.

Suppose by contradiction that $\varepsilon > \delta$.
Since~(2) is false, there exists a $u\in B_\delta(x) \setminus B_\varepsilon(y)$.
But then we have
\begin{equation}
\begin{split}
d(u,y)
& \leq
\max\{ d(u,z) , d(z,y) \}
\leq
\max\{ d(u,x) , d(x,z) , d(z,y) \}
\\ & <
\max\{ \delta , \delta , \varepsilon \}
=
\varepsilon
\enspace,
\end{split}
\end{equation}
which means that $u\in B_\varepsilon(y)$, a contradiction to the choice of~$u$.
\end{proof}

The \emph{morphisms} of topological spaces $X$ and $Y$ are \emph{continuous} functions, 
namely, those functions $f:X\to Y$ such that any inverse image of an open set is open.
In metric terms, this means that for every $x\in X$ and every $\varepsilon > 0$ there exists a $\delta > 0$ such that $d_X(x,x') < \delta$ implies $d_Y(f(x),f(x')) < \varepsilon$ for all $x'\in X$.
Here, we denoted by~$d_X$ the metric on~$X$ and by~$d_Y$ the metric on~$Y$.
All constant functions are continuous, as are all \emph{locally constant} functions, 
\ie, functions $f:X\to Y$ that are constant in some open ball $B_\varepsilon(x)$ 
with positive radius $\varepsilon>0$ for every $x\in X$.

Topologies for standard set-theoretic constructions can be defined from their individual parts.
For instance, the product topology of a countable collection of metric spaces~$X_i$ can be described by the metric $d:X\times X\to[0,\infty)$ with 
\begin{equation}
d(x,y)
=
\sum_{i\in \mathbb{N}} 2^{-i} \frac{d_i(x_i, y_i)}{1 + d_i(x_i,y_i)}
\enspace.
\end{equation}

We use the product metric to extend the notion of indistinguishability of local views of configurations
(these concepts are formally defined in Section~\ref{sec:ll-model}) to a metric on infinite executions.
It has the following property:

\begin{lem}[{\cite[\S~2.3, Proposition~4]{bourbaki:general:topology:1:4}}]\label{lem:prod:top}
Let $(X_i)_{i\in \mathbb{N}}$ be a countable collection of metric spaces and let $X = \prod_{i\in \mathbb{N}} X_i$ be their product equipped with the product metric.
For all metric spaces~$Y$ and all functions $g:Y\to X$, the following are equivalent:
\begin{enumerate}
\item The function~$g$ is continuous.
\item The function $\pi_i\circ g$ is continuous for all $i\in \mathbb{N}$ where $\pi_i:X\to X_i$ is the projection on the component~$i$.
\end{enumerate}
\end{lem}

The disjoint-union topology of the disjoint union $X = \bigsqcup_{i\in I} X_i$ is described by the metric $d:X\times X\to[0,\infty)$ with $d(x,y) = d_i(x,y)$ if there is an index $i\in I$ such that both~$x$ and~$y$ are elements of~$X_i$, and $d(x,y) = 2$ else.
We use the disjoint-union metric to get a global metric from those defined for the local views of each process,
with the following property:

\begin{lem}[{\cite[\S~2.4, Proposition~6]{bourbaki:general:topology:1:4}}]\label{lem:union:top}
Let $(X_i)_{i\in I}$ be a collection of metric spaces and let $X = \bigsqcup_{i\in I} X_i$ be their disjoint union equipped with the disjoint-union metric.
For all metric spaces~$Y$ and all functions $g:X\to Y$, the following are equivalent:
\begin{enumerate}
\item The function~$g$ is continuous.
\item The function $g\circ\varphi_i$ is continuous for all $i\in I$ where $\varphi_i:X_i\to X$ is the embedding of~$X_i$ into~$X$.
\end{enumerate}
\end{lem}

\subsection{System Model}
\label{sec:ll-model}

Let $\task = (\inputs, \outputs, \Delta)$ any task.
Since our goal is to give a very general characterization of task solvability, we work with an abstract system model that hides most of the operational details, such as semantics of shared registers or guarantees of message delivery.
We instead focus on the structure of the set of executions induced by the local \emph{indistinguishability} relations, \ie, by the processes' local views.
We further assume that actions taken by processes do not influence the set of possible executions.
That is, we assume the existence of \emph{full-information} executions, on which we base our characterization.
A full-information execution is a sequence of configurations.
A configuration is a vector with the  process states and the state of the environment (\eg, shared memory, messages in transit) in its entries.
In a full-information execution, every process relays all the information it gathered to all other processes whenever it can.
This includes its input value, the order and contents of events it perceived, and the information relayed to it by others.
In particular, we assume that there are no size constraints on messages or shared memory.

Formally,
let~$\exec$ be the set of full-information executions of~$n$ processes in which initial configurations are chosen according to the input complex~$\inputs$.
We assume the existence of projection functions $\pi_p : \exec \to \view_p$ from executions to sequences of local views of process~$p$.
These sequences can be finite or infinite.
Its element with index~$t$ contains the local view of process~$p$ right after 
its $t$\textsuperscript{th} step in the execution.
The set of process views of executions in which process~$p$ is correct will be denoted by $\cview_p$.

A step is defined as a possibility to irrevocably decide.
That is, the $t$\textsuperscript{th} step of process~$p$ is process~$p$'s $t$\textsuperscript{th} possibility to decide a value (or not) in the execution.
We allow processes to decide in their initial state, \ie, in their step with index $t=0$.
Step counts are local to a process and need not be synchronized among processes.
A process that only has finitely many steps is called \emph{faulty} in the execution.
For an execution $E\in\exec$ we write $\correct(E)\subseteq \Pi$ for the set of correct (non-faulty) processes in the execution.
A \emph{participating} process $p\in\participating(E)\subseteq \Pi$ is one that takes at least one step.
We have that $\correct(E)\subseteq\participating(E)$.
We write $\init_p(E)\in V^{\mathrm{in}}$ for the initial value of process~$p$ in execution~$E$.
The concrete forms of executions and local views depend on the specifics of the computational model.
\figref{fig:execs} depicts an example execution and process-view sequence.

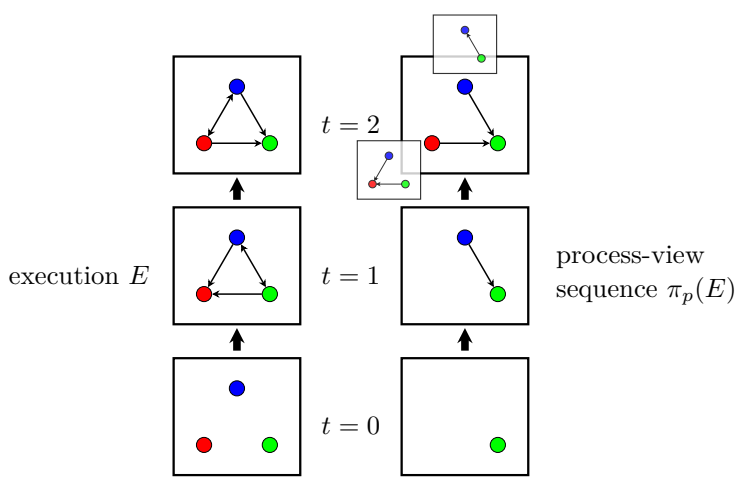
\begin{figure}
\centering
\begin{tikzpicture}[
    state/.style={circle, draw=black, inner sep=0pt, minimum size=5pt, fill=gray}
]
\begin{scope}[
    xshift=0mm, yshift=0mm,
    scale=1, transform shape,
    every node/.append style={opacity=1.0, scale=1, transform shape},
    opacity=1.0
]
\draw[line width=0.3mm,fill=white] (-0.8330127018922193,-0.65) rectangle (0.8330127018922193,0.9);
\node[state,fill=red, line width=0.1mm, minimum size=2.0mm] (t0g1) at (-0.4330127018922193,-0.25) {};
\node[state,fill=blue, line width=0.1mm, minimum size=2.0mm] (t0g2) at (0.0,0.5) {};
\node[state,fill=green, line width=0.1mm, minimum size=2.0mm] (t0g3) at (0.4330127018922193,-0.25) {};
\end{scope}
\draw[-{Stealth[length=2mm, width=2mm]}, line width=1mm] (0mm, 10.0mm) -- (0mm, 13.0mm);
\begin{scope}[
    xshift=0mm, yshift=20mm,
    scale=1, transform shape,
    every node/.append style={opacity=1.0, scale=1, transform shape},
    opacity=1.0
]
\draw[line width=0.3mm,fill=white] (-0.8330127018922193,-0.65) rectangle (0.8330127018922193,0.9);
\node[state,fill=red, line width=0.1mm, minimum size=2.0mm] (t1g1) at (-0.4330127018922193,-0.25) {};
\node[state,fill=blue, line width=0.1mm, minimum size=2.0mm] (t1g2) at (0.0,0.5) {};
\node[state,fill=green, line width=0.1mm, minimum size=2.0mm] (t1g3) at (0.4330127018922193,-0.25) {};
\draw[-{Stealth[length=1.0mm, width=1.0mm]}, line width=0.2mm] (t1g2) -- (t1g1);
\draw[-{Stealth[length=1.0mm, width=1.0mm]}, line width=0.2mm] (t1g3) -- (t1g1);
\draw[{Stealth[length=1.0mm, width=1.0mm]}-{Stealth[length=1.0mm, width=1.0mm]}, line width=0.2mm] (t1g3) -- (t1g2);
\end{scope}
\draw[-{Stealth[length=2mm, width=2mm]}, line width=1mm] (0mm, 30.0mm) -- (0mm, 33.0mm);
\begin{scope}[
    xshift=0mm, yshift=40mm,
    scale=1, transform shape,
    every node/.append style={opacity=1.0, scale=1, transform shape},
    opacity=1.0
]
\draw[line width=0.3mm,fill=white] (-0.8330127018922193,-0.65) rectangle (0.8330127018922193,0.9);
\node[state,fill=red, line width=0.1mm, minimum size=2.0mm] (t2g1) at (-0.4330127018922193,-0.25) {};
\node[state,fill=blue, line width=0.1mm, minimum size=2.0mm] (t2g2) at (0.0,0.5) {};
\node[state,fill=green, line width=0.1mm, minimum size=2.0mm] (t2g3) at (0.4330127018922193,-0.25) {};
\draw[{Stealth[length=1.0mm, width=1.0mm]}-{Stealth[length=1.0mm, width=1.0mm]}, line width=0.2mm] (t2g2) -- (t2g1);
\draw[{Stealth[length=1.0mm, width=1.0mm]}-, line width=0.2mm] (t2g3) -- (t2g1);
\draw[{Stealth[length=1.0mm, width=1.0mm]}-, line width=0.2mm] (t2g3) -- (t2g2);
\end{scope}
\begin{scope}[
    xshift=30mm, yshift=0mm,
    scale=1, transform shape,
    every node/.append style={opacity=1.0, scale=1, transform shape},
    opacity=1.0
]
\draw[line width=0.3mm,fill=white] (-0.8330127018922193,-0.65) rectangle (0.8330127018922193,0.9);
\node[state,fill=green, line width=0.1mm, minimum size=2.0mm] (t0p3) at (0.4330127018922193,-0.25) {};
\end{scope}
\draw[-{Stealth[length=2mm, width=2mm]}, line width=1mm] (30mm, 10.0mm) -- (30mm, 13.0mm);
\begin{scope}[
    xshift=30mm, yshift=20mm,
    scale=1, transform shape,
    every node/.append style={opacity=1.0, scale=1, transform shape},
    opacity=1.0
]
\draw[line width=0.3mm,fill=white] (-0.8330127018922193,-0.65) rectangle (0.8330127018922193,0.9);
\node[state,fill=blue, line width=0.1mm, minimum size=2.0mm] (t1p2) at (0.0,0.5) {};
\node[state,fill=green, line width=0.1mm, minimum size=2.0mm] (t1p3) at (0.4330127018922193,-0.25) {};
\draw[{Stealth[length=1.0mm, width=1.0mm]}-, line width=0.2mm] (t1p3) -- (t1p2);
\end{scope}
\draw[-{Stealth[length=2mm, width=2mm]}, line width=1mm] (30mm, 30.0mm) -- (30mm, 33.0mm);
\begin{scope}[
    xshift=30mm, yshift=40mm,
    scale=1, transform shape,
    every node/.append style={opacity=1.0, scale=1, transform shape},
    opacity=1.0
]
\draw[line width=0.3mm,fill=white] (-0.8330127018922193,-0.65) rectangle (0.8330127018922193,0.9);
\node[state,fill=red, line width=0.1mm, minimum size=2.0mm] (t2p1) at (-0.4330127018922193,-0.25) {};
\node[state,fill=blue, line width=0.1mm, minimum size=2.0mm] (t2p2) at (0.0,0.5) {};
\node[state,fill=green, line width=0.1mm, minimum size=2.0mm] (t2p3) at (0.4330127018922193,-0.25) {};
\draw[{Stealth[length=1.0mm, width=1.0mm]}-, line width=0.2mm] (t2p3) -- (t2p1);
\draw[{Stealth[length=1.0mm, width=1.0mm]}-, line width=0.2mm] (t2p3) -- (t2p2);
\end{scope}
\begin{scope}[
    xshift=20.0mm, yshift=33.333333333333336mm,
    scale=0.5, transform shape,
    every node/.append style={opacity=0.8, scale=0.5, transform shape},
    opacity=0.8
]
\draw[line width=0.15mm,fill=white] (-0.8330127018922193,-0.65) rectangle (0.8330127018922193,0.9);
\node[state,fill=red, line width=0.05mm, minimum size=4.0mm] (t2pq1) at (-0.4330127018922193,-0.25) {};
\node[state,fill=blue, line width=0.05mm, minimum size=4.0mm] (t2pq2) at (0.0,0.5) {};
\node[state,fill=green, line width=0.05mm, minimum size=4.0mm] (t2pq3) at (0.4330127018922193,-0.25) {};
\draw[-{Stealth[length=0.5mm, width=0.5mm]}, line width=0.1mm] (t2pq2) -- (t2pq1);
\draw[-{Stealth[length=0.5mm, width=0.5mm]}, line width=0.1mm] (t2pq3) -- (t2pq1);
\end{scope}
\begin{scope}[
    xshift=30mm, yshift=50.0mm,
    scale=0.5, transform shape,
    every node/.append style={opacity=0.8, scale=0.5, transform shape},
    opacity=0.8
]
\draw[line width=0.15mm,fill=white] (-0.8330127018922193,-0.65) rectangle (0.8330127018922193,0.9);
\node[state,fill=blue, line width=0.05mm, minimum size=4.0mm] (t2pr2) at (0.0,0.5) {};
\node[state,fill=green, line width=0.05mm, minimum size=4.0mm] (t2pr3) at (0.4330127018922193,-0.25) {};
\draw[-{Stealth[length=0.5mm, width=0.5mm]}, line width=0.1mm] (t2pr3) -- (t2pr2);
\end{scope}
\node at (15.0mm, 0mm) {$t=0$};
\node at (15.0mm, 20mm) {$t=1$};
\node at (15.0mm, 40mm) {$t=2$};
\node[text width=24.0mm] at (-18.0mm, 20mm) {execution~$E$};
\node[text width=24.0mm] at (54.0mm, 20mm) {process-view sequence $\pi_p(E)$};
\end{tikzpicture}
\caption{Prefix of a full-information execution~$E$ (left) and process-view projection~$\pi_p(E)$ (right) of a synchronous message-passing system with dynamic communication graphs. The depicted prefix includes the initial configuration as well as the first two communication rounds. Process~$p$ is the green (lower right) process. Initially, after round~$0$, process~$p$ only knows its own initial value. After the first round, process~$p$ also knows the blue (upper) process's initial value as well as the fact that directed edge from the blue to the green process was present in the communication graph of the first round. After the second round, process~$p$ learned the initial value of the red (lower left) process, its own incoming edges of the second round's communication graph, as well as the views of the blue and the red processes after the first round.}
\label{fig:execs}
\end{figure}

A (decision) \emph{protocol} is a function from local views to $V^{\mathrm{out}}\cup\{\perp\}$ with $\perp\not\in V^{\mathrm{out}}$ such that decisions are irrevocable:
if some view is mapped to a decision value $v\in V^{\mathrm{out}}$, then all its successor views are also mapped to~$v$.
A process~$p$ thus has at most one decision value in every execution~$E$, which we denote by $\decision_p(E) \in V^{\mathrm{out}}$.
A protocol \emph{solves} a task $T=(\inputs,\outputs,\Delta)$ if 
it satisfies the following two conditions in every execution $E\in\exec$:
\begin{itemize}
\item Every correct process $p\in\correct(E)$ has a decision value in~$E$.
\item We have $\{ v(p,\decision_p(E)) \mid p\in \correct(E) \} \in \Delta(\sigma)$ where $\sigma = \{ v(p,\init_p(E)) \mid p\in\participating(E) \}$.
\end{itemize}

\smallskip\noindent{\bf Example: Lossy-Link Model.}
The lossy-link model~\cite{SW89} is a synchronous computation model with $n=2$ processes, $p$ and $q$, that communicate via message passing.
The communication graph can change from round to round.
In each round, the adversary chooses one of three communication graphs:
$\leftarrow$, $\rightarrow$, or $\leftrightarrow$.
In a round with communication graph~$\leftarrow$, only the message from the right to the left processes arrives, the other message is lost.
In a round with communication graph~$\rightarrow$, only the message from the left to the right processes arrives, the other message is lost.
In a round with communication graph~$\leftrightarrow$, both messages arrive and no message is lost.
In a full-information execution, each process starts out by sending its initial value and then records all received messages in subsequent rounds, relaying this information to the other process.
In this model, there is no notion of faulty processes; the only source of uncertainty is the communication.
We thus have $\participating(E) = \correct(E) = \{p,q\} = \Pi$ for every full-information execution~$E$,
and thus $\cview_p = \view_p$ and $\cview_q = \view_q$.

Since both processes are correct in every execution, both processes are participating, \ie, $\participating(E) = \Pi$.



\section{The Need of Continuity}
\label{sec:need-conitnuity}

This section explains the need of continuity of simplicial maps to model
protocols in non-compact models. This is done using 
in part the example in~\cite{GP20} showing 
a flaw in the attempt to generalize the ACT~\cite{GKM14}. 

For a system with two processes, left and right, 
the compact IIS model can be equivalently defined as the 
lossy-link model described in Section~\ref{sec:ll-model}. 
Thus, IIS for two processes
consists of all infinite sequences of communication graphs 
$\leftarrow$, $\rightarrow$, or $\leftrightarrow$, each graph
specifying the communication that occurs in a round. 
A sub-IIS model is any subset of IIS.

Let us consider an \emph{inputless} version of the consensus task 
where the left process has fixed input~$0$ and 
the right process has fixed input~$1$.
Then, the input complex of the task $\inputs$ is the complex made of 
the edge $\sigma = \{0, 1\}$ and its faces (processes are identified with their inputs),
the output complex~$\outputs$
has simplexes $\{0\}$ and $\{1\}$, and $\Delta$ maps $\sigma$ to $\outputs$,
and each $\{ i \}$ to itself. 
Complexes $\inputs$ and $\outputs$ will be denoted 
$\sigma$ and $\partial \sigma$, respectively.

The topology of the IIS executions is well understood:
the complex modeling all configurations at the end of round $R$ is a \emph{finite} subdivision of 
the input complex $\sigma$ (basically a subdivision of the real interval $[0,1]$), 
and as $R$ increases, the subdivision gets finer.
Concretely, it is the $R$-th standard chromatic subdivision.
Figure~\ref{fig:IIS}(left) shows the subdivisions for the first two rounds, 
where, for example, the left-most and right-most edges of the second subdivision 
correspond to 
the configurations at the end of the finite solo executions 
$\rightarrow, \rightarrow$ and $\leftarrow, \leftarrow$, respectively, 
where a process does not hear from the other,
and the central edge corresponds to $\leftrightarrow, \leftrightarrow$,
where processes hear from each other.

\begin{figure}[tb]
\begin{center}
\includegraphics[scale=0.65]{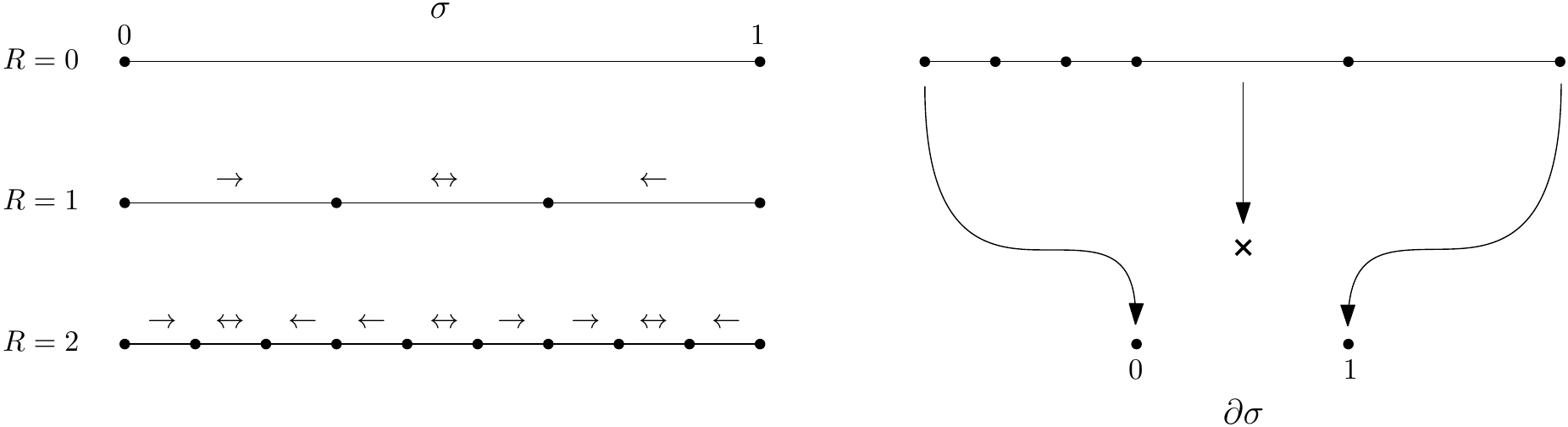}
\caption{Subdivisions in the IIS model for two processes.}
\label{fig:IIS}
\end{center}
\end{figure}

A key property of round-structured compact models like IIS is that, 
for any protocol solving a task,
there is a \emph{finite} round $R$ such that all correct processes make a decision
at round $R$, at the latest (assuming $\inputs$ is finite). 
With this property, it is simple to see that consensus is impossible in~IIS
(see the right side of Figure~\ref{fig:IIS}):
\begin{enumerate}
\item For any round $R$, the complex corresponding to the \emph{decided states} of 
a hypothetical protocol~$P$, is a finite subdivision $\m{K}$ of $\sigma$, 
\ie, $|\m{K}| = |\sigma|$. 
(Recall that $|\m{K}|$ is the geometric realization of $\m{K}$.)
The subdivision might be irregular because processes might make decisions at different 
rounds; processes keep running after decision, hence an edge models infinitely many infinite executions, 
all of them sharing the finite prefix where the decisions are made.

\item $P$ must map each vertex (state) of $\m{K}$ to an output in $\partial \sigma$,
with the restriction that the left-most vertex must be mapped to $0$ and the
right-most vertex must be mapped to~$1$, as they correspond to solo executions,
hence, by validity of the consensus task 
(\ie, $\Delta(\{i\}) = \{i\}$), 
the process that only sees its input is forced to decide it.

\item Since $P$ solves consensus, it induces a simplicial map
$\delta: \m{K} \rightarrow \partial \sigma$, which, as $\m{K}$ is finite, 
\emph{necessarily} induces a continuous map
$|\delta|: |\m{K}| \rightarrow |\partial \sigma|$. 
The map $|\delta|$ is ultimately a continuous map
$|\sigma| \rightarrow |\partial \sigma|$ that maps the boundary of $\sigma$ to itself.

\item Finally, this continuous map does not exist because $|\sigma|$ is solid whereas
$|\partial \sigma|$ is disconnected.
\end{enumerate}

The argument above goes from protocols to simplicial maps.
In models like IIS, the other direction is also true. 
Namely, for any given task $\task = (\inputs, \outputs, \Delta)$,
for any complex $\m{K}$ related to $\inputs$ that satisfies some model-dependent properties,
any simplicial map from $\m{K}$ to $\outputs$ 
that agrees with $\Delta$, induces a protocol for $\task$.
Thus, to show that a task is solvable in two-process IIS, 
it suffices to exhibit a finite, possibly irregular, subdivision of the input complex, 
in the style of the one in Figure~\ref{fig:IIS}(left),
and a simplicial map that is valid for the~task.

The main aim of~\cite{GKM14} is to generalize the approach above that equates
simplicial maps and protocols to arbitrary sub-IIS models, in order to exploit the
already known topology of IIS. The high-level idea is that the complexes that model
a sub-IIS model are still subdivisions but not necessarily of the input complex,
and not necessarily finite. 

Let us consider first the sub-IIS model $M_1$ with all infinite executions
of the form $\leftarrow$ followed by any infinite sequence with $\leftarrow$, $\rightarrow$, 
or $\leftrightarrow$ (intuitively right goes first),
or $\rightarrow$ followed by any infinite sequence with $\leftarrow$, $\rightarrow$, 
or $\leftrightarrow$ (intuitively left goes first).
It can be seen that consensus is solvable in this model: 
since $\leftrightarrow$ cannot happen in the
first round, the process that receives no message in the first round is the ``winner''.
Figure~\ref{fig:sub-IIS-1} shows an irregular subdivision that models all executions of $M_1$;
for example, the right-most edge corresponds to all executions of $M_1$ with prefix
$\leftarrow, \leftarrow$. Intuitively, in the subdivision,
in some executions processes decide in round one (represented by the edge at the left), 
and in the remaining executions processes decide in round two 
(represented by the three edges at the right).
Clearly, there is a simplicial map from such a disconnected subdivision 
to $\partial \sigma$ that agrees with consensus.
This simplicial map induces a consensus protocol for $M_1$.

\begin{figure}[tb]
\begin{center}
\includegraphics[scale=0.7]{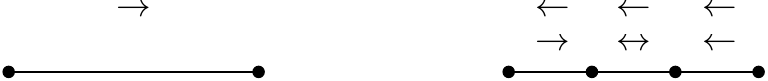}
\caption{A possible subdivision for the sub-IIS model $M_1$.}
\label{fig:sub-IIS-1}
\end{center}
\end{figure}

The argument above works well because the model is compact, hence finite subdivisions 
are able to capture all its executions. 
However, in non-compact models,
some executions can only be modeled through infinite subdivisions, 
which implies that simplicial maps are not necessarily protocols.

Consider now the sub-IIS model $M_2$ obtained by removing from IIS
the infinite execution~$E$ described by the sequence $\leftrightarrow,\leftarrow, \leftarrow, \hdots$
This model is not compact because it contains any infinite execution with a
\emph{finite} prefix (of any length) of~$E$, 
but it does not contain~$E$ itself, the limit execution.
As said, a crucial property of non-compactness is that the executions of the model
cannot be captured by a finite subdivision. Intuitively, an edge can only model 
executions that have a common finite prefix of~$E$ of length $x$, 
but in $M_2$ there are executions with
a prefix larger than $x$, 
hence these executions are not captured by the edge;
if the subdivision is finite, there are necessarily executions that are not modeled by any edge.

Figure~\ref{fig:sub-IIS-2} contains a schematized
infinite subdivision $\m{K}$ that indeed captures all executions of~$M_2$.
Intuitively, there are infinitely many edges that get closer and closer
to the point that represents the removed 
execution~$E$ 
(depicted as a vertical dashed line at the center),
but no edge actually ``crosses'' it (as $E$ is not in $M_2$).
Thus the simplicial complex \m{K} is
disconnected, and there is a simplical map from \m{K} to $\partial \sigma$ 
that agrees with consensus. Although all executions are captured
in the infinite subdivision, such a  simplicial map \emph{does not} imply a protocol. 
The intuition is that there is a sudden jump in the decisions around $E$,
which ultimately implies that the decision in executions that are similar enough to the removed 
limit execution~$E$ are not consistent, namely, they cannot be produced
by a protocol.

\begin{figure}[tb]
\begin{center}
\includegraphics[scale=0.65]{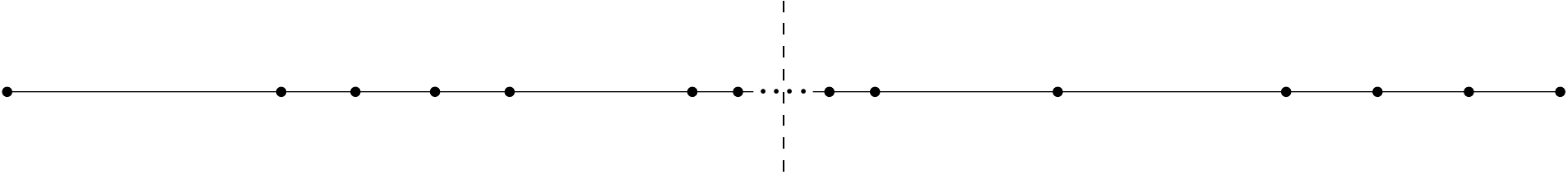}
\caption{An schematic representation of an infinite subdivision for the sub-IIS model $M_2$.}
\label{fig:sub-IIS-2}
\end{center}
\end{figure}

It turns out that the topological space $|\m{K}|$ is actually a subdivision of $|\sigma|$: 
in the limit, $|\m{K}| = |\sigma|$.
Thus, the infinite subdivision \m{K} describes a space that is not disconnected!
Moreover, for any infinite subdivision that models $M_2$, 
the space associated with it is connected, 
\ie, this is an invariant of the model of computation. 
Any simplicial map that intends to capture 
a protocol should consider that $|\m{K}| = |\sigma|$.
This is precisely captured by demanding that 
the induced map $|\m{K}| \rightarrow |\outputs|$ must be continuous (hence smooth around $E$).
Therefore, there is no continuous map
$|\sigma| \rightarrow |\partial \sigma|$ that maps the boundary of $\sigma$ to itself,
and indeed consensus is not solvable in this model~\cite[Theorem~III.8]{FevatG2011minimal}.

Formalizing this seemingly simple observation 
in arbitrary models of computation is not obvious, 
and it requires a combination of 
combinatorial topology techniques and point-set topology techniques, 
as is done in the following sections.
Intuitively, distance functions in point-set topology
are used to equip protocol complexes with a topology
that in turn yields a correspondence between
continuous simplicial maps and protocols.

\section{Proof of the Generalized Asynchronous Computability Theorem with an Application to Set Agreement}

In this section we use the definitions and notation of Gafni, Kuznetsov, and Manolescu~\cite{GKM14} for sub-IIS models.
They introduced the notion of \emph{terminating subdivisions} of the input complex~$\inputs$ of a task.
The idea is to repeatedly subdivide all simplexes via the standard chromatic subdivision, except those that are already marked as terminated.
The terminated simplexes model configurations where processes have decided.

Formally, a terminating subdivision~$\terminating$ 
is specified by a sequence of chromatic complexes 
$\m{I}_0, \m{I}_1, \hdots$ and a sequence of
subcomplexes $\Sigma_0 \subseteq \Sigma_1 \subseteq \hdots$
such that for all $k \geq 0$:
(1) $\Sigma_k$ is a subcomplex of $\m{I}_k$ 
(each $\m{I}_{k}$ is a \emph{non-uniform} subdivision~\cite{HS06})
and
(2) $\m{I}_0 = \m{I}$ and $\m{I}_{k+1}$ is obtained from $\m{I}_k$
by the \emph{partial} chromatic subdivision in which
the simplexes in $\Sigma_k$ are not further subdivided (the terminated simplexes), and each simplex $\tau \notin \Sigma_k$
is replaced with its standard chromatic subdivision $\Chr \tau$.
Precisely, we replace a simplex $\sigma$ in $\m{I}_k$ by a coarser 
subdivision than $\Chr \sigma$. Whereas the vertices of $\Chr \sigma$ are pairs 
$(p,\sigma')$ with $p \in \Pi$ and $\sigma' \subseteq \sigma$,
in $\m{I}_{k+1}$ we consider the pairs $(p,\sigma')$ of that form such that either $\sigma' \notin \Sigma_k$, or $\sigma'$ consists of a single vertex in $\Sigma_k$.

Figure~\ref{fig:term-sub} schematizes a terminating subdivision
where $\inputs$ is made of two triangles
and terminated simplexes are marked in red.

A simplex of $\Sigma_k$, for some $k$, is called \emph{stable}.
The simplicial complex~$K(\terminating)$ is the union of all~$\Sigma_k$;
$K(\terminating)$ might be infinite.

\begin{figure}[tb]
\begin{center}
\includegraphics[scale=0.45]{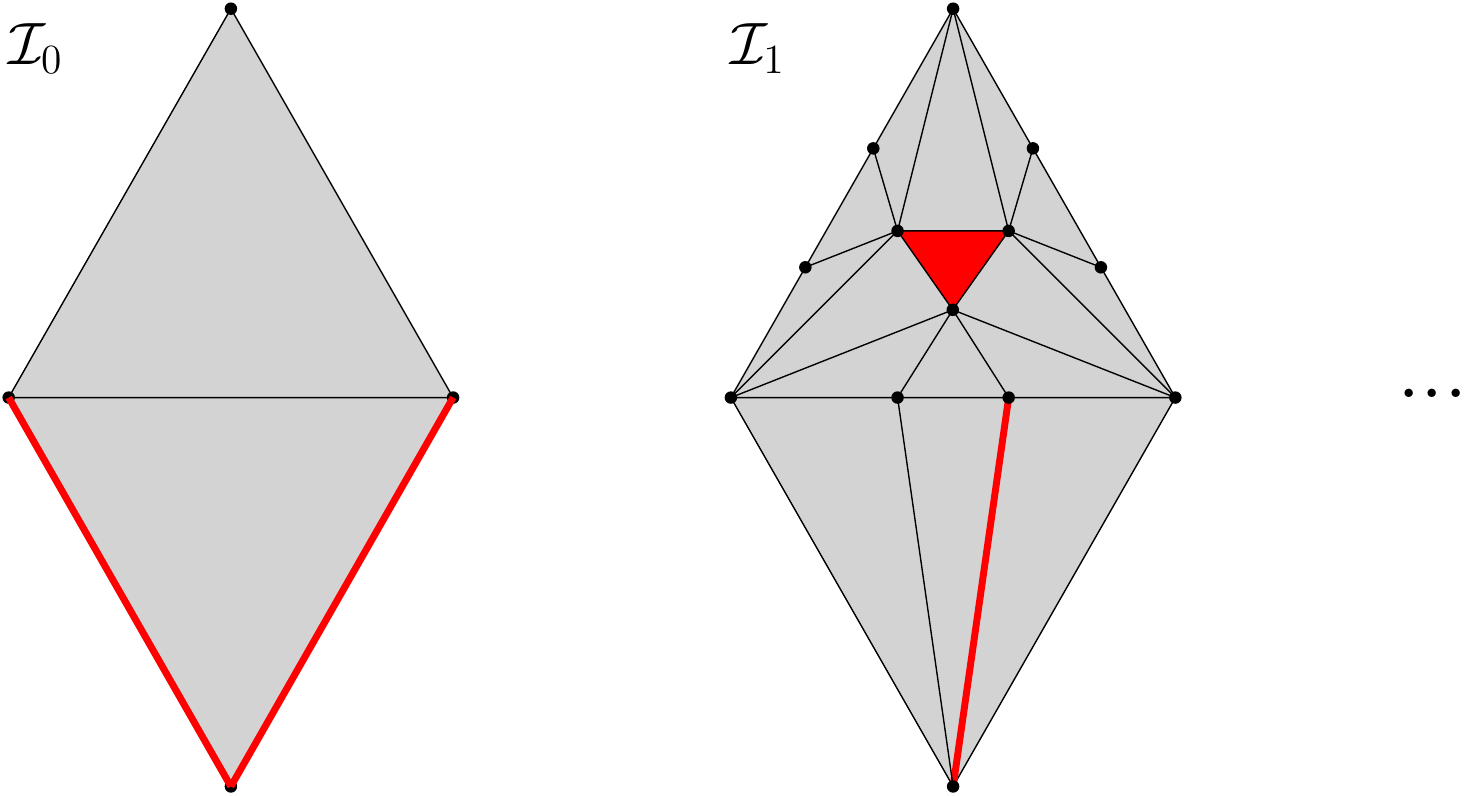}
\caption{First two complexes of 
a three-process terminating subdivision.}
\label{fig:term-sub}
\end{center}
\end{figure}

The vertices of~$K(\terminating)$ are naturally embedded in the geometric realization of~$\inputs$ by their definition as a vertex of the repeated chromatic subdivision~$\Chr^k \inputs$ 
(recall that $|\Chr^k \inputs| = |\inputs|$, for every $k \geq 0$).
In particular, we identify the geometric realization~$\lvert K(\terminating)\rvert$ with a subset of~$\lvert \inputs\rvert$.
Every IIS execution can be described as an infinite
sequence of simplexes $\sigma_0, \sigma_1, \hdots$
such that $\sigma_k \in \Chr^k \m{I}$, for every $k \geq 0$.

A terminating subdivision is \emph{admissible} for a sub-IIS model~$M$ if $K(\terminating)$ covers all executions of~$M$,
namely, for each execution $\sigma_0, \sigma_1, \hdots$ of $M$,
there is a $k$ such that $|\sigma_k| \subseteq |\tau|$,
for some terminated simplex $\tau \in \Sigma_k$.

\begin{thmrep}\label{thm:GACT-fixed}
A sub-IIS model~$M$ solves a task $\task = (\inputs, \outputs, \Delta)$ if and only if there exists a terminating subdivision~$\terminating$ of~$\inputs$ and a chromatic simplicial map $\delta \colon K(\terminating)\to\outputs$ such that:
\begin{enumerate}[(a)]
\item $\terminating$ is admissible for the model~$M$.
\item For any simplex~$\sigma$ of~$\inputs$, if~$\tau$ is a stable 
simplex of~$\terminating$ such that $\lvert \tau\rvert \subseteq \lvert \sigma\rvert$, then $\delta(\tau)\in\Delta(\sigma)$.
\item $\lvert \delta\rvert$ is continuous.
\end{enumerate}
\end{thmrep}
\begin{proof}
($\Rightarrow$):
We prove that the geometric realization of the map~$\delta$ as constructed in the proof of Gafni, Kuznetsov, and Manolescu \cite[Theorem~6.1]{GKM14} is continuous by generalizing the proof given by Godard and Perdereau~\cite[Theorem~33]{GP20} for the consensus task with two processes.

Let $x\in \lvert K(\terminating)\rvert$ and~$\varepsilon>0$.
We show the existence of an $\eta > 0$ such that:
\begin{equation}\label{eq:gact:delta:cont:def}
\forall y\in \lvert K(\terminating)\rvert
\colon\quad
d(x,y) < \eta
\implies
d\big( \lvert\delta\rvert(x) , \lvert\delta\rvert(y) \big) < \varepsilon
\end{equation}
Let~$\sigma$ be the minimal stable simplex in~$K(\terminating)$ such that $x\in\lvert\sigma\rvert$.
Since~$K(\terminating)$ is locally finite, the star $\sta \sigma = \{\tau\in K(\terminating) \mid \sigma\subseteq\tau \}$ is finite.
Let~$k$ be the smallest round number such that $\sta \sigma \subseteq \Sigma_k$.
Denote by~$D_k$ the diameter of the geometric realization of simplices in $\Chr^k \inputs$
and choose $\eta = \varepsilon D_k$.

We show~\eqref{eq:gact:delta:cont:def} in the geometric realization of every simplex $\tau\in\sta\sigma$.
By the choice of~$k$, we have $\tau\in \Chr^r \inputs$ for some $0\leq r\leq k$.
Let $y\in\lvert\tau\rvert$ and denote by~$\alpha$ the barycentric coordinates of~$x$ with respect to~$\tau$ and by~$\beta$ the barycentric coordinates of~$y$ with respect to~$\tau$, \ie,
$x = \sum_{v\in\tau} \alpha(v) \cdot v$ and $y = \sum_{v\in\tau} \beta(v) \cdot v$ with $\alpha,\beta\geq0$ and $\lVert\alpha\rVert_1 = \lVert\beta\rVert_1 = 1$.
Here, we identified each vertex~$v\in\tau$ with its position in the geometric realization~$\lvert K(\terminating)\rvert$.
We then have:
\begin{equation}
d(x,y)
=
\lVert x - y\rVert_1
=
\diam\lvert \tau \rvert \cdot \sum_{v\in\tau} \lvert \alpha(v) - \beta(v)\rvert
\geq
D_k \cdot \sum_{v\in\tau} \lvert \alpha(v) - \beta(v)\rvert
\end{equation}
By definition of the geometric realization~$\lvert\delta\rvert$, we have
\begin{equation}
\lvert\delta\rvert(x) = \sum_{v\in \tau} \alpha(v) \cdot \delta(v)
\end{equation}
where, again, we identify the vertex~$\delta(v)$ with its position in geometric realization~$\lvert\outputs\rvert$.
Since~$\delta(v)$ is a vertex of~$\outputs$ for every vertex $v\in\tau$, we have
\begin{equation}
\begin{split}
d\big( \lvert\delta\rvert(x) , \lvert\delta\rvert(y) \big)
\leq
\sum_{v\in\tau} \lvert \alpha(v) - \beta(v)\rvert
\leq
\frac{d(x,y)}{D_k}
<
\frac{\eta}{D_k}
=
\varepsilon
\enspace,
\end{split}
\end{equation}
which shows~\eqref{eq:gact:delta:cont:def} and concludes the proof of continuity of~$\lvert\delta\rvert$.

\medskip

($\Leftarrow$):
We modify the protocol that is constructed in the proof of Gafni, Kuznetsov, and Manolescu \cite[Theorem~6.1]{GKM14} for process~$p$ to decide in round~$k$ if the set
\begin{equation}
B_k(v) = \{ w \in V(K(\terminating)) \mid d(v,w) \leq D_k \wedge \chi(w) = p \}
\end{equation}
only contains vertices that are mapped to the same output vertex by~$\delta$, where~$v$ is the view of the process in round~$k$.
This condition eventually becomes true since the subset topology on the geometric realization of the output vertices~$V(\outputs)$ is discrete, and thus~$\delta$ is locally constant.
\end{proof}

\begin{proofsketch}
($\Rightarrow$):
This direction consisting in showing that the geometric realization of the map~$\delta$ as constructed in the proof of Gafni, Kuznetsov, and Manolescu \cite[Theorem~6.1]{GKM14} is continuous by generalizing the proof given by Godard and Perdereau~\cite[Theorem~33]{GP20} for the consensus task with two processes.

\medskip

($\Leftarrow$):
We modify the protocol that is constructed in the proof of Gafni, Kuznetsov, and Manolescu \cite[Theorem~6.1]{GKM14} for process~$p$ to decide in round~$k$ if the set
\begin{equation}
B_k(v) = \{ w \in V(K(\terminating)) \mid d(v,w) \leq D_k \wedge \chi(w) = p \}
\end{equation}
only contains vertices that are mapped to the same output vertex by~$\delta$, where~$v$ is the view of the process in round~$k$.
This condition eventually holds since the subset topology on the 
geometric realization of the output vertices~$V(\outputs)$ is discrete;
thus,~$\delta$ is locally constant.
\end{proofsketch}

We now use Theorem~\ref{thm:GACT-fixed} to derive a condition for
the impossibility of $(n-1)$-set agreement task in IIS-sub models. 
Recall that in this task each process is required to eventually decide an input value (termination) 
of a process participating in the execution (validity) such that no more 
than $n-1$ distinct values are decided (agreement). 

Let $\Pi = \{p_0, p_1, \hdots, p_{n-1}\}$.
For simplicity, we focus on the \emph{inputless} version of the set agreement task, 
where each process
$p_i$, $0 \leq i \leq n-1$, has fixed input $i$ in every execution,
and thus the task is the triple $\task = (\inputs, \outputs, \Delta)$,
where the input complex $\inputs$ is made of all faces of simplex $\sigma = \{0, 1, \hdots, n-1\}$,
and for simplicity it is denoted $\sigma$,
the output complex~$\outputs$, denoted $\partial \sigma$, is the complex with all \emph{proper} faces of $\sigma$,
and $\Delta$ maps every proper face $\sigma' \subset \sigma$ to
the complex with all faces of $\sigma'$, and maps $\sigma$ to~$\partial \sigma$.

\begin{thmrep}
\label{theo-set-agreement}
Let $M$ be an IIS-sub model such that for any termination subdivision 
$\terminating$ of~$\sigma$ that is admissible for $M$, 
$\lvert \sigma \rvert = \lvert K(\terminating) \rvert$. 
Then, $(n-1)$-set agreement is impossible in $M$.
\end{thmrep}
\begin{proof}
Let $M$ be a sub-IIS model. 
By Theorem~\ref{thm:GACT-fixed}, 
if $(n-1)$-set agreement is solvable in model~$M$,
there is a (possibly infinite) terminating subdivision $\terminating$ of $\sigma$ 
and a chromatic simplicial map $\delta: K(\terminating) \rightarrow \partial \sigma$
such that 
(1) $\terminating$ is admissible for $M$,
(2) for every input simplex~$\sigma' \subseteq \sigma$, if~$\tau$ is a stable simplex of~$\terminating$
 such that $\lvert \tau\rvert \subseteq \lvert \sigma'\rvert$, then $\delta(\tau)\in\Delta(\sigma')$,
and (3)~$\lvert \delta\rvert$ is continuous.

Let us suppose that $\lvert \sigma \rvert = \lvert K(\terminating) \rvert$, namely, 
$K(\terminating)$ subdivides $\sigma$.
Thus, for each face $\sigma' \subseteq \sigma$, $\lvert \sigma' \rvert = \lvert K(\sigma') \rvert$,
where $K(\sigma')$ denotes the terminating subdivision of $\sigma'$.
Consider the identity map $g: \lvert \sigma \rvert \rightarrow \lvert K(\terminating) \rvert$.
Clearly, $g$ is continuous, with $g(\lvert \sigma' \rvert) = \lvert K(\sigma') \rvert$.
Consider the function $f = \lvert \delta \rvert \circ g : \lvert \sigma \rvert \rightarrow \lvert \partial \sigma \rvert$.
Since $\lvert \delta \rvert$ and $g$ are continuous, the function $f$ is continuous too. 
We argue that $f(\lvert \sigma' \rvert)~\subseteq~\lvert \sigma' \rvert$,
for every proper face $\sigma' \subset \sigma$.
Consider any proper face $\sigma' \subset \sigma$.
We have that 
(a) $g(\lvert \sigma' \rvert) = \lvert K(\sigma') \rvert$, by definition of $g$,
(b) for any stable simple $\tau \in \terminating$ with  
$\lvert \tau\rvert \subseteq \lvert \sigma'\rvert = \lvert K(\sigma') \rvert$,
$\delta(\tau)\in\Delta(\sigma')$, by the properties of $\delta$, and
(c) $\Delta(\sigma') = \sigma'$, by definition of $\Delta$.
We thus conclude that $f(\lvert \sigma' \rvert)~\subseteq~\lvert \sigma' \rvert$.

The following lemma is direct consequence of Lemma 4.3.5 in~\cite{HKR-book}, 
and proves below the impossibility of $(n-1)$-set agreement whenever 
$\lvert K(\terminating) \rvert = \lvert \sigma \rvert$.

\begin{lem}
\label{lemma-sperner-continuous}
There is no continuous map $f : \lvert \sigma \rvert \rightarrow \lvert \partial \sigma \rvert$
such that for every proper face $\sigma' \subset \sigma$, 
$f(\lvert \sigma' \rvert)~\subseteq~\lvert \sigma' \rvert$.
\end{lem}

One can understand Lemma~\ref{lemma-sperner-continuous} 
as a continuous version of the discrete Sperner's lemma. 
Intuitively, it states that if a continuous map 
$f : \lvert \sigma \rvert \rightarrow \lvert \partial \sigma \rvert$ maps 
the boundary of $\sigma$ to itself 
(\ie, $f(\lvert \sigma' \rvert) \subseteq \lvert \sigma' \rvert$,
for each $\sigma' \subset \sigma$), 
similar to Sperner's lemma hypothesis, then $f$ cannot exist because
the mapping cannot be extended to the interior of $\sigma$,
since $\lvert \sigma \rvert$ is solid whereas $\lvert \partial \sigma \rvert$
has a hole.

As explained above, Theorem~\ref{thm:GACT-fixed} 
and assumption $\lvert K(\terminating) \rvert = \lvert \sigma \rvert$ imply that 
if $(n-1)$-set agreement is solvable in $M$,
then there exists a continuous map 
$f : \lvert \sigma \rvert \rightarrow \lvert \partial \sigma \rvert$
such that for every proper face $\sigma' \subset \sigma$, 
$f(\lvert \sigma' \rvert)~\subseteq~\lvert \sigma' \rvert$.
Such continuous map $f$ contradicts Lemma~\ref{lemma-sperner-continuous}.
Therefore, $(n-1)$-set agreement is impossible in $M$.
\end{proof}

\section{Characterization of Task Solvability in General Models}\label{sec:topology}

In this section we present a topological 
solvability characterization in general models, hence showing that 
the topology approach is applicable in all models of computation. 
As anticipated, the characterization demands simplicial maps
to be continuous, which is particularly relevant if complexes are infinite. 
Differently from sub-IIS models in the previous section,
where continuity naturally arises in protocols complexes as they
are subdivisions (hence embedded in a Euclidean space), 
the general case requires to equip protocol 
complexes with a topology, which is used to capture continuity.

Recall that the set of process views of executions in which process~$p$ 
is correct is denoted $\cview_p$.
These are the executions in which we demand process~$p$ to decide on an output value.
We always have $\cview_p \subseteq \view_p$.
For every process~$p$ we define a topology on the set~$\cview_p$
of correct process-$p$ local-view sequences induced by the distance function
\begin{equation}
d_p(\alpha, \beta)
=
2^{-T_p(\alpha,\beta)}
\end{equation}
where $T_p(\alpha,\beta)$ is defined as the smallest index at which the local views in the process-view sequences~$\alpha$ and~$\beta$ differ.
If no such index exists, then $T_p(\alpha,\beta)=\infty$.
This means that the distance between two process-view sequences is smaller the later the process can detect a difference between the two.
If~$\alpha$ and~$\beta$ do not differ in any index, then $\alpha=\beta$ and $d_p(\alpha,\beta) = 2^{-\infty} = 0$.
A variant of this distance function, which considers complete executions instead of local views, was introduced by Alpern and Schneider~\cite{AS84}.

We first establish that the distance function~$d_p$ is an ultrametric.

\begin{lemrep}
The distance function~$d_p$ is an ultrametric on~$\cview_p$.
\end{lemrep}
\begin{proof}
If $\alpha=\beta$, then $d_p(\alpha,\beta)=0$.
If $d_p(\alpha,\beta)=0$, then $T_p(\alpha,\beta)=\infty$, which means that there is no index at which they differ by definition, \ie, $\alpha = \beta$.
This shows that~$d_p$ is positive definite.

The symmetry condition $d_p(\alpha,\beta)=d_p(\beta,\alpha)$ holds since the definition of $T_p(\alpha,\beta)$ is symmetric in~$\alpha$ and~$\beta$.

We now prove the ultrametric triangle inequality by showing:
\begin{equation}\label{eq:dist:ultrametric}
T_p(\alpha,\gamma) \geq \min\big\{ T_p(\alpha,\beta) , T_p(\beta,\gamma) \big\}
\end{equation}
Assume by contradiction that $T_p(\alpha,\gamma) < \min\big\{ T_p(\alpha,\beta) , T_p(\beta,\gamma) \big\}$.
Set $t = T_p(\alpha,\gamma)$.
Since $t < \infty$ and $t < T_p(\alpha,\beta)$, all local views up to index~$t$ coincide in both sequences~$\alpha$ and~$\beta$.
Likewise, all local views up to index~$t$ coincide in both sequences~$\beta$ and~$\gamma$.
But then, by transitivity of the equality relation on local views, all local views up to index~$t$ coincide also in the two sequences~$\alpha$ and~$\gamma$, which means $T_p(\alpha,\gamma) > t = T_p(\alpha,\gamma)$; a contradiction.
\end{proof}

In the next lemma, we establish the fundamental fact that the decision functions for process~$p$ are exactly the continuous functions $\cview_p \to V(\outputs)$ when using~$d_p$ on~$\cview_p$ and the discrete metric on~$V(\outputs)$.

\begin{lemrep}\label{lem:char:cont}
Let $\delta_p:\cview_p\to V(\outputs)$ be a function.
The following are equivalent:
\begin{enumerate}
\item There is a protocol such that process~$p$ decides the value~$\delta_p(\alpha)$ in every execution $E\in\exec$ with local view $\pi_p(E)=\alpha\in\cview_p$.
\item The function~$\delta_p$ is continuous when equipping $\cview_p$ with the topology induced by~$d_p$ and~$V(\outputs)$ with the discrete topology.
\end{enumerate}
\end{lemrep}

\begin{proof}
($\Rightarrow$):
To show that~$\delta_p$ is continuous, we will show that the inverse image of any singleton $\{o\}\subseteq V(\outputs)$ is open with respect to~$d_p$.
This then implies that the inverse image of any subset $O\subseteq V(\outputs)$, \ie, of any subset of~$V(\outputs)$ that is open with respect to the discrete topology, is open with respect to~$d_p$.

Let $\alpha \in \delta_p^{-1}[\{o\}]$.
Because process~$p$ decides the value~$o$ in the local view~$\alpha$, there exists an index~$T$ at which this decision has already happened in~$\alpha$.
Choose $\varepsilon = 2^{-T}$.
Now let $\alpha'\in\cview_p$ with $d_p(\alpha,\alpha') < \varepsilon$.
Then, by definition of~$d_p$, the local views of~$\alpha$ and of~$\alpha'$ are indistinguishable for process~$p$ up to and including index~$T$.
But then, process~$p$ needs to have decided value~$o$ at index~$T$ in local view~$\alpha'$ as well.
We thus have $\delta_p(\alpha') = o$, which means that $\alpha' \in \delta_p^{-1}[\{o\}]$.
Therefore, the inverse image $\delta_p^{-1}[\{o\}]$ is open with respect to~$d_p$.

($\Leftarrow$):
We define the protocol for process~$p$ in the following way.
Decide value $o\in V(\outputs)$ in the $t$\textsuperscript{th} step if the set of all local views in $\cview_p$ that are indistinguishable from the current execution in the first~$t$ steps of process~$p$ is included in the inverse image $\delta_p^{-1}[\{o\}]$.

Let $E\in\exec$ be an execution with $p\in\correct(E)$.
We will show that process~$p$ decides value $o=\delta_p(\alpha)$ where $\alpha = \pi_p(E)$.
By definition of~$o$, we have $\alpha \in \delta_p^{-1}[\{o\}]$.
By continuity of~$\delta_p$, the inverse image $\delta_p^{-1}[\{o\}]$ is an open set in $\cview_p$.
There hence exists an $\varepsilon > 0$ such that $\alpha'\in \delta_p^{-1}[\{o\}]$ for all $\alpha'\in\cview_p$ with $d_p(\alpha,\alpha')<\varepsilon$.
It remains to show that process~$p$ eventually decides the value~$o$ and that it does not decide any other value in execution~$E$.
Setting $T = \lceil \log_2 \varepsilon \rceil$, we see that, by design of the protocol's decision rule, process~$p$ has decided value~$o$ at the latest in step number~$T$.
To show that process~$p$ does not decide any other value than~$o$, it suffices to observe that $d(\alpha,\alpha)=0<2^{-t}$ for every $t\geq 0$ and $\alpha \in \delta_p^{-1}[\{o\}]$.
\end{proof}

To formulate and prove our characterization for the solvability of tasks in general models, we define a structure that combines the notions of chromatic simplicial complexes and the notion of point-set topology of sequences of local views.
Formally, a \emph{topological chromatic simplicial complex} is a chromatic simplicial complex whose set of vertices is equipped with a topology.
A vertex map between two topological chromatic simplicial complexes is a \emph{morphism} if it is continuous, chromatic, and simplicial.

The protocol complex~$\protocol$ is a (possibly infinite) topological chromatic simplicial complex defined as follows.
The set of vertices of~$\protocol$ is the disjoint union $V(\protocol) = \bigsqcup_{p\in\Pi} \cview_p$ of the correct local-view spaces.
The vertices from~$\cview_p$ are colored with the process name~$p$.
We equip the set of vertices with the disjoint-union topology, \ie, the finest topology that makes all embedding maps $\iota_p:\cview_p\to V(\protocol)$ continuous.

A set~$\sigma$ of vertices of~$\protocol$ is a simplex of~$\protocol$ if and only if the local views are consistent with the views of correct processes in an execution, \ie, if it is of the form
\begin{equation}
\sigma
=
\left\{
\pi_p(E) \mid p\in P
\right\}
\end{equation}
for some execution $E \in \exec$ and some set $P\subseteq \correct(E)$.

The execution map $\Xi: \inputs \to 2^\protocol$ is defined by mapping every input simplex in~$\inputs$ to the local views of correct processes of executions in which the initial values of participating processes are as in the input simplex.
Formally,
\begin{equation}
\Xi(\sigma)
=
\big\{
\{ \pi_p(E) \mid p\in P \}
\mid
E\in \compatible(\sigma)
\wedge 
P\subseteq \correct(E)
\big\}
\end{equation}
where $\compatible(E)\subseteq\exec$ denotes the set of executions that are compatible with the initial values described by the input simplex~$\sigma$.
That is,
\begin{equation}
\compatible(\sigma)
=
\left\{
E\in\exec \mid 
\participating(E) \subseteq \chi[\sigma]
\wedge
\forall p\in\participating(E)\colon \init_p(E) = \ell(\pi_p(\sigma))
\right\}
\end{equation}
where~$\pi_p(\sigma)$ denotes the unique vertex of the simplex~$\sigma$ with the color~$p$, if it exists,
and~$\ell(\pi_p(\sigma))$ is the input (label) of
vertex~$\pi_p(\sigma)$. 
The execution map~$\Xi$ assigns a subcomplex of~$\protocol$ to every input simplex~$\sigma$ in~$\inputs$.
As in the classical finite-time setting~\cite[Definition~8.4.1]{HKR-book}, 
the next lemma shows that it is a carrier map:

\begin{lemrep}
The execution map\/~$\Xi$ is a carrier map such that $\displaystyle \protocol = \bigcup_{\sigma\in\inputs} \Xi(\sigma)$.
\end{lemrep}
\begin{proof}
We first prove that~$\Xi$ is a carrier map.
Let $\sigma\subseteq\tau$.
We need to prove that $\Xi(\sigma)\subseteq \Xi(\tau)$.
We first show that $\compatible(\sigma)\subseteq \compatible(\tau)$.
Let $E\in \compatible(\sigma)$.
Then $\participating(E) \subseteq \chi[\sigma] \subseteq \chi[\tau]$ since $\sigma\subseteq \tau$, which means that $E\in \compatible(\tau)$ because the second condition in the definition is fulfilled since $\pi_p(\sigma)=\pi_p(\tau)$ for all $p\in\chi[\sigma]$.
But then, for every $\varphi \in \Xi(\sigma)$, we also have $\varphi \in \Xi(\tau)$ since every~$E$ in the definition of~$\Xi(\sigma)$ is also valid for~$\Xi(\tau)$.

To prove $\protocol = \bigcup_{\sigma\in\inputs} \Xi(\sigma)$, it suffices to show $\exec = \bigcup_{\sigma\in\inputs}\compatible(\sigma)$.
The inclusion of $\compatible(\sigma)$ in~$\exec$ is immediate by its definition.
So let $E\in\exec$ be any execution.
We need to show the existence of a simplex $\sigma\in\inputs$ such that $E\in\compatible(\sigma)$.
For this, it suffices to choose $\sigma = \{ v(p,\init_p(E)) \mid p\in\participating(E) \}$.
This set is an input simplex since the initial values in~$\exec$ are chosen according to~$\inputs$ by definition.
\end{proof}

In contrast to the classical finite-time setting, however,
the execution map is not necessarily rigid.
Whether it is depends on whether any finite execution prefix can be extended to a fault-free execution.
This is not the case, \eg, in many synchronous models.
If~$\Xi$ is not rigid, then, by definition, it is \emph{a fortiori} not chromatic.
It does, however, satisfy the inclusion 
\begin{equation}
\big\{ \chi(v) \mid v\in V(\Xi(\sigma)) \big\}
\ \subseteq\ 
\chi[\sigma]
\end{equation}
for all input simplices $\sigma\in\inputs$.
In other words, the colors of~$\Xi(\sigma)$ are included in the colors of~$\sigma$; no new process names appear.
It turns out that the stronger assumptions of rigidity or chromaticity are not necessary to show our solvability characterization.

\begin{thm}\label{thm:topology-approach}
The task $\task = (\inputs, \outputs, \Delta)$ is solvable if and only if there exists a continuous chromatic simplicial map $\delta:\protocol\to\outputs$ such that $\delta\circ \Xi$ is carried by~$\Delta$.
\end{thm}
\begin{proof}
($\Rightarrow$):
Assume that there is a protocol that solves task~$T$.
Define the vertex map $\delta:\protocol\to\outputs$ by setting $\delta(\alpha)$ to be the vertex of~$\outputs$ with color~$p$ and label~$v$ where~$p$ is the unique process such that $\alpha\in\cview_p$ and~$v$ is the decision value of process~$p$ in an execution with local-view sequence~$\alpha$ when executing the protocol.

The map~$\delta$ is continuous on each individual~$\cview_p$ by Lemma~\ref{lem:char:cont}.
By Lemma~\ref{lem:union:top}, it is thus continuous on their disjoint union~$\protocol$.
The map~$\delta$ is chromatic since the color of the vertex $\alpha \in\cview_p$ is~$p$, as is the color of~$\delta(\alpha)$.

To prove that~$\delta$ is simplicial, let~$\varphi$ be a simplex of~$\protocol$.
Then, by definition, there exists an execution $E\in\exec$ and a set $P\subseteq \correct(E)$ such that $\varphi = \{ \pi_p(E) \mid p\in P\}$.
Set $\sigma = \{ v(p,\init_p(E)) \mid p\in\Pi \}$ and $\tau = \{ v(p,\decision_p(E) \mid p\in\correct(E) \}$.
Then, since the protocol solves task~$T$, we have $\tau\in \Delta(\sigma)$.
By definition of~$\delta$, we then have $\delta[\varphi] = \{ v(p,\decision_p(E)) \mid p\in P \} \subseteq \tau \in\Delta(\sigma)\subseteq \outputs$, which means that $\delta[\varphi]\in\outputs$ and hence that~$\delta$ is simplicial.

It remains to prove that~$\delta\circ\Xi$ is carried by~$\Delta$.
So let~$\sigma \in \inputs$ and $\tau \in (\delta\circ\Xi)(\sigma)$.
We need to show that $\tau\in\Delta(\sigma)$.
By the definitions of~$\Xi$ and~$\delta$, there exists an execution $E\in\compatible(\sigma)$ and a set $P\subseteq\correct(E)$ such that  $\tau = \{v(p,\decision_p(E)) \mid p\in P\}$.
Since the protocol solves task~$T$, we have $\tau' = \{v(p,\decision_p(E))\mid p\in\correct(E)\} \in \Delta(\sigma')$ where $\sigma' = \{v(p,\init_p(E)) \mid p\in\participating(E) \}$.
Since $\tau\subseteq\tau'$ and~$\Delta(\sigma')$ is a simplicial complex, 
we deduce that $\tau\in\Delta(\sigma')$.
Now, because $E\in \compatible(\sigma)$, we have $\participating(E) \subseteq \chi[\sigma]$ and $\sigma' \subseteq \sigma$.
It thus follows that $\tau \in \Delta(\sigma') \subseteq \Delta(\sigma)$ because~$\Delta$ is a carrier map.

\medskip

($\Leftarrow$):
The restriction~$\delta_p$ of~$\delta$ to the set $\cview_p$ is continuous because~$\delta$ is.
By Lemma~\ref{lem:char:cont} there hence exists a protocol such that every process~$p$ decides the value $\ell(\delta(\pi_p(E))) \in V^{\mathrm{out}}$ for every execution~$E$ in which~$p$ is correct.

Let $E\in\exec$ be any execution and define the sets $\sigma = \{v(p,\init_p(E)) \mid p\in \participating(E)\}$ and $\tau = \{v(p,\decision_p(E)) \mid p\in\correct(E)\}$.
To show that the protocol solves task~$T$, it remains to show that $\tau\in\Delta(\sigma)$.
Since $\delta\circ\Xi$ is carried by~$\Delta$, it suffices to prove $\tau \in (\delta\circ\Xi)(\sigma)$.
Setting $\varphi = \{ \pi_p(E) \mid p\in\correct(E) \}$, we have $\tau = \delta[\varphi]$.
We are thus done if we show $\varphi \in \Xi(\sigma)$.
But this follows from $E\in\compatible(\sigma)$, which is true by construction of~$\sigma$.
\end{proof}

\section{Relationship to the Classical Finite-Time Approach}

In this section, we formalize the relationship between our infinite protocol complex used for the general solvability characterization in Theorem~\ref{thm:topology-approach} and the classically studied finite-time protocol complexes.
Besides demonstrating that the classical formalism is a special case of ours, we show the finite-time approach is sufficient for all compact models.
More specifically, we show that it is possible to restrict the study to finite-time protocols if the computational model is compact.
Formally, a topological space is \emph{compact} if every open cover has a finite subcover.
Many computational models that are defined by safety predicates are compact.
We use the concept of projective limit from category theory~\cite{MacLane78} to formalize the relationship between finite-time and infinite-time complexes.
In particular, we show that the infinite-time complex is the projective limit of the finite-time complexes if the model is compact.

\smallskip\noindent{\bf Finite-Time Complexes.}
For every nonnegative integer~$T$, we define the time-$T$ protocol complex $\timeprotocol{T}$ as follows:
\begin{itemize}
\item The set of vertices of $\timeprotocol{T}$ is the disjoint union of the sets $\timecview{p}{T}$ where~$p$ varies in the set~$\Pi$ of processes.

\item The set $\timecview{p}{T}$ is defined as the set of open balls of radius $\varepsilon=2^{-T}$ in $\cview_p$.
These balls are either identical or disjoint by Lemma~\ref{lem:ultrametric:balls}.

\item All vertices of $\timecview{p}{T}$ are colored with~$p$.

\item A set of vertices of $\timeprotocol{T}$ is a simplex of $\timeprotocol{T}$ if and only if there is a simplex of~$\protocol$ that is formed by choosing one element in each vertex of the set.

\item The topology on $V(\timeprotocol{T})$ is the discrete topology.
\end{itemize}

This definition makes $\timeprotocol{T}$ a topological chromatic simplicial complex.
As a chromatic simplicial complex, it is isomorphic to the classical finite-time construction of protocol complexes~\cite{HKR-book}.
The finite-time execution map $\Xi_T:\inputs\to 2^{\timeprotocol{T}}$ is defined by
\begin{equation}
\Xi_T(\sigma)
=
\big\{
B_T[\tau]
\ \big|\ 
\tau\in \Xi(\sigma)
\big\}
\end{equation}
where $B_T(\alpha) = \{\beta\in V(\protocol) \mid d(\alpha,\beta) < 2^{-T}\}$ is the function that takes each vertex~$\alpha$ of~$\protocol$ to the open $2^{-T}$-ball in which it is included.

\smallskip\noindent{\bf Projective Limits.}
We will show that, if the model is compact, then~$\protocol$ is the limit of the~$\timeprotocol{T}$ in a precise sense.
For this, we use the notion of projective limits from category theory~\cite{MacLane78}, which we introduce in this subsection.

A \emph{category} is a class of \emph{objects} and a class of \emph{morphisms} between objects.
Every morphism $f:X\to Y$ is assigned a domain object~$X$ and a codomain object~$Y$.
For compatible morphisms $f:X\to Y$ and $g:Y\to Z$, the composition $g\circ f$ is a morphism $X\to Z$.
The composition operator is required to be associative.
For every object~$X$, the existence of an identity morphism $\id_X:X\to X$ is required.
The identity morphism satisfies $f \circ \id_x = f$ for all morphism $f:X\to Y$ with domain~$X$ and $\id_X \circ g = g$ for all morphisms $g:Z\to X$ with codomain~$X$.

A sequence $(X_T)_{T\geq0}$ of objects of a category can be transformed into an \emph{inverse system} by specifying a family $(f_{S,T})_{0\leq S\leq T}$ of morphisms $f_{S,T}:X_T\to X_S$ such that $f_{T,T} = \id_{X_T}$ and $f_{R,T} = f_{R,S} \circ f_{S,T}$ for all $0\leq R\leq S\leq T$.
The \emph{projective limit} of the sequence is then an object~$X$ together with morphisms $\pi_T:X\to X_T$ such that $\pi_S = f_{S,T}\circ \pi_T$ for all $0\leq S\leq T$ and with the universal property that for any other such object~$Y$ and morphisms $\psi_T: Y\to X_T$, there exists a unique morphism $u:Y\to X$ such that the following diagram commutes for all $0\leq S\leq T$:

\begin{equation}\label{eq:cd:proj:limit}
\begin{tikzcd}
		& Y \arrow[ldd, "\psi_T"'] \arrow[d, "u"] \arrow[rdd, "\psi_S"] &     \\
		& X \arrow[ld, "\pi_T"] \arrow[rd, "\pi_S"']                    &     \\
		X_T \arrow[rr, "f_{S,T}"'] &                                                               & X_S
\end{tikzcd}
\end{equation}

For every pair of integers $S$ and $T$, $0\leq S\leq T$, define the vertex maps 
$f_{S,T} \colon \timeprotocol{T} \to \timeprotocol{S}$ by setting $f_{S,T}(B)$ to be the unique open $2^{-S}$-ball of $\timeprotocol{S}$ in which the open $2^{-T}$-ball~$B$ of $\timeprotocol{T}$ is included.
These are morphisms between topological chromatic simplicial complexes and they satisfy $f_{R,T} = f_{R,S}\circ f_{S,T}$ for all $0\leq R\leq S\leq T$.
This makes the sequence of the~$\timeprotocol{T}$ an inverse system.

\begin{lem}\label{lem:proj:limit:exists}
The projective limit of the sequence of complexes $\timeprotocol{T}$ exists.
\end{lem}
\begin{proof}
From the category of sets
we borrow the usual limit construction
\begin{equation}\label{eq:def:proj:limit:vertices}
V(\mathcal{L})
=
\left\{
B \in \prod_{T\geq 0} V(\timeprotocol{T})
\ \Big|\ 
\forall 0\leq S\leq T\colon
B_S = f_{S,T}(B_T)
\right\}
\end{equation}
with the vertex-map projections
$\pi_T:\mathcal{L}\to\timeprotocol{T},\ \pi_T(B) = B_T$
for the set of vertices of the purported limit complex~$\mathcal{L}$.
We have $\pi_S = f_{S,T} \circ \pi_T$ for all $0\leq S\leq T$ by construction since $f_{S,T}(\pi_T(B)) = f_{S,T}(B_T) = B_S = \pi_S(B)$ for all vertices~$B$ of~$\mathcal{L}$.

We equip the set~$V(\mathcal{L})$ with the product topology.
We define a set $\sigma\subseteq V(\mathcal{L})$ to be a simplex of~$\mathcal{L}$ if and only if every projection~$\pi_T[\sigma]$ of~$\sigma$ is a simplex of~$\timeprotocol{T}$.
The color of a vertex~$B$ of~$\mathcal{L}$ is the unique color of all the~$\pi_T(B)$.
Uniqueness of this color follows from the fact that the~$f_{S,T}$ are chromatic, which implies that colors cannot change in a sequence $B\in V(\mathcal{L})$ since $B_S = f_{S,T}(B_T)$ means that~$B_S$ and~$B_T$ have the same color.
These definitions make~$\mathcal{L}$ a topological chromatic simplicial complex and all the projections~$\pi_T$ morphisms in this category.

It remains to prove the universal property of the projective limit.
So let~$\mathcal{K}$ be a topological chromatic simplicial complex and let $\psi_T:\mathcal{K}\to\timeprotocol{T}$ be morphisms such that $\psi_S = f_{S,T}\circ\psi_T$ for all $0\leq S\leq T$.
We need to show the existence of a unique morphism $u:\mathcal{K}\to\mathcal{L}$ such that $\psi_T = \pi_T\circ u$ for all~$T$.
We set $u(k) = \big( \psi_T(k) \big)_{T\geq0}$ for every vertex~$k$ of~$\mathcal{K}$.
We have $\pi_T(u(k)) = \psi_T(k)$ for every~$k$ by construction and the definition of the projection~$\pi_T$.
This also implies uniqueness of~$u$.
Continuity of~$u$ follows from Lemma~\ref{lem:prod:top} and the continuity of the~$\psi_T$.
To show that~$u$ is simplicial, let~$\sigma$ be a simplex of~$\mathcal{K}$.
Since~$\psi_T$ is simplicial, the set $\psi_T[\sigma]$ is a simplex of~$\timeprotocol{T}$ for every $T\geq0$.
But then we deduce that $\pi_T[u[\sigma]] = (\pi_T\circ u)[\sigma] = \psi_T[\sigma]$ is a simplex of~$\timeprotocol{T}$ for all~$T\geq0$, which means that the set~$u[\sigma]$ is a simplex of~$\mathcal{L}$ by definition.
To show that~$u$ is chromatic, we use the fact that~$\pi_T$ is chromatic to deduce that the vertex~$u(k)$ has the same color as  $\pi_T(u(k)) = (\pi_T\circ u)(k) =\psi_T(k)$, which, in turn, has the same color as~$k$ since~$\psi_T$ is chromatic.
This shows the universal property of the limit complex~$\mathcal{L}$ and concludes the proof of the lemma.
\end{proof}

We can equip the set of executions with the metric
$d(E,E') = 2^{-K}$ where $K = \inf\{k\geq0\mid E_k\neq E_k'\}$,
which measures how many configurations are identical in two execution prefixes~\cite{AS84}.
With this topology on~$\exec$, the projection maps $\pi_p:\exec\to\view_p$ are continuous.
In fact, continuity of the map means that each local view needs to be determined by some finite prefix of the execution.
We have the following lemma:

\begin{lem}
If\/~$\exec$ is compact, then~$V(\protocol)$ is compact as well, and $\protocol$ is the projective limit of the~$\timeprotocol{T}$.
\end{lem}
\begin{proof}
The sets~$\cview_p$ are compact as continuous images of a compact set.
But then~$V(\protocol)$ is compact as a finite union of compact sets.

To prove the second statement, we use the limit complex~$\mathcal{L}$ as defined in the proof of Lemma~\ref{lem:proj:limit:exists}.
We show that~$\protocol$ is isomorphic to~$\mathcal{L}$ by defining morphisms $f:\protocol\to\mathcal{L}$ and $g:\mathcal{L}\to\protocol$ such that $f\circ g = \id_{\mathcal{L}}$ and $g\circ f = \id_{\protocol}$.

We define the vertex maps $f(\alpha) = (B_T(\alpha))_{T\geq 0}$ and $g(B) = \alpha$ where~$\alpha$ is the unique element of the limit set $\lim_{T\to\infty} B_T = \bigcap_{T\geq0} B_T$.
The map~$g$ is well-defined since~$V(\protocol)$ is compact:
For every~$T\geq0$, the sets~$B_T(\alpha)$ with~$\alpha$ varying in~$V(\protocol)$ form an open covering of~$V(\protocol)$.
Since these open $2^{-T}$-balls are either disjoint or equal by Lemma~\ref{lem:ultrametric:balls}, the complement of~$B_T$ is open as a union of open balls.
But this means that~$B_T$ is a closed set.
Since every closed subset of a compact space is compact, so is~$B_T$.
The limit of any descending sequence of compact sets is nonempty, which in particular implies that $\lim_{T\to\infty} B_T \neq \emptyset$.
Since $\diam(B_T) \leq 2^{-T}\to 0$ as $T\to\infty$, the limit set cannot contain more than one element; it thus contains exactly one.

We have $f(g(B)) = B$ for all $B\in V(\mathcal{L})$ since $g(B)\in B_T$ for all $T\geq 0$, which means that $B_T(g(B)) = B_T$.
We also have $g(f(\alpha)) = \alpha$ for all $\alpha\in V(\protocol)$ since $\alpha \in B_T(\alpha)$ for all $T\geq0$, which means that $\alpha \in \bigcap_{T\geq0} B_T(\alpha)$.
It remains to show that~$f$ and~$g$ are morphisms of topological chromatic simplicial complexes.

The map~$f$ is continuous by Lemma~\ref{lem:prod:top} since the function $\pi_T\circ f : V(\protocol)\to V(\timeprotocol{T}), (\pi_T\circ f)(\alpha) = B_T(\alpha)$ is locally constant, and hence continuous, for every~$T\geq0$.
In fact, if $d(\alpha,\beta) < 2^{-T}$, then $B_T(\alpha) = B_T(\beta)$.
To show that the map~$g$ is continuous,
we use the formula $d(B, B') = 2^{-T}$ where $T = \inf\{t\geq 0\mid B_t\neq B_t' \}$
for the metric of the product of discrete spaces
(\eg, \cite[Lemma~4.7]{NSW19:podc}).
Using this formula, we see that $d(g(B),g(B')) \leq d(B,B')$, which proves that~$g$ is continuous.

To show that~$f$ is simplicial, let $\sigma \in \protocol$ be a simplex.
Then, by definition of the complex~$\timeprotocol{T}$, the set~$B_T[\sigma]$ is a simplex of~$\timeprotocol{T}$ for all $T\geq0$ and thus, by definition, the set~$f[\sigma]$ is a simplex of~$\mathcal{L}$.
To show that~$g$ is simplicial, let $\sigma \in \mathcal{L}$ be a simplex.
Then, by definition, the set~$\pi_T[\sigma]$ is a simplex of~$\timeprotocol{T}$ for every~$T\geq0$.
Setting $P = \chi[\sigma]$, there hence exist executions $E^{(T)}\in\exec$ for all $T\geq0$ such that $P\subseteq \correct(E)$ and $\pi_p(E^{(T)})\in \pi_T(\pi_p(\sigma))$ for all~$p\in P$.
Since~$\exec$ is compact, the sequence of the~$E^{(T)}$ has a convergent subsequence, \ie, there is an increasing sequence of indices~$T_s$ and some $E\in\exec$ such that $d(E^{(T_s)}, E) \to 0$ as $s\to\infty$.
The limit execution~$E$ of this subsequence satisfies $\pi_p(E)\in \pi_T(\pi_p(\sigma))$ for all $p\in P$ since the~$\pi_T(\pi_p(\sigma))$ are compact and hence closed.
But then, since there is only one element in the intersection over $T\geq0$, we have $\pi_p(E) = g(\pi_p(\sigma))$.
Thus, $g[\sigma] = \{ \pi_p(E) \mid p\in P\}$ with $P\subseteq\correct(E)$, which means that~$g[\sigma]$ is a simplex of~$\protocol$.

To show that~$f$ is chromatic, we note that every~$B_T(\alpha)$ has the same color as~$\alpha$ since 
$d(\alpha,\beta)=2 > 2^{-T}$ for any~$\beta$ that has a different color than~$\alpha$.
By the same argument, $g$ is chromatic as well.
\end{proof}

\smallskip\noindent{\bf Sufficiency of Finite-Time Complexes for Compact Models.}
We can now state the fact that finite-time protocol complexes are sufficient to study compact models.

\begin{thm}\label{thm:projective:limit:compact}
If\/~$V(\protocol)$ is compact, then the following are equivalent:
\begin{enumerate}
\item The task $T = (\inputs,\outputs,\Delta)$ is solvable.
\item There is a continuous chromatic simplicial map $\delta:\protocol\to\outputs$ such that $\delta\circ\Xi$ is carried by~$\Delta$.
\item There is a time~$T$ such that there exists a chromatic simplicial map $\delta_T:\timeprotocol{T}\to\outputs$ such that $\delta_T\circ\Xi_T$ is carried by~$\Delta$.
\item The task $T = (\inputs, \outputs, \Delta)$ is solvable in a bounded number of local steps per process.
\end{enumerate}
\end{thm}
\begin{proof}
We will show the implication chain (1)$\Rightarrow$(2)$\Rightarrow$(3)$\Rightarrow$(4)$\Rightarrow$(1).The implication (1)$\Rightarrow$(2) is included in Theorem~\ref{thm:topology-approach} and the implication (4)$\Rightarrow$(1) is trivial.
It thus remains to prove the implications (2)$\Rightarrow$(3) and (3)$\Rightarrow$(4).

To prove (2)$\Rightarrow$(3), we first note that the image $D=\delta[V(\protocol)]\subseteq V(\outputs)$ of the decision map is finite since it is compact as the continuous image of a compact set and the only compact sets in a discrete space are the finite sets.
The inverse images of elements of~$D$ thus form a finite partition of the set of vertices of~$\protocol$:
\begin{equation}
V(\protocol)
=
\bigsqcup_{o\in D} \delta^{-1}[\{o\}]
\end{equation}
Since the singleton sets~$\{o\}$ are closed in the discrete topology, the inverse images $\delta^{-1}[\{o\}]$ are closed by continuity of~$\delta$.
Being closed subsets of a compact space, they are also compact.
Two disjoint compact sets in a metric space have a strictly positive distance:
\begin{equation}
\exists \varepsilon_{o,o'} > 0 \colon\quad
d\big( \delta^{-1}[\{o\}] \ ,\ \delta^{-1}[\{o'\}] \big)
=
\inf_{\substack{\alpha\in \delta^{-1}[\{o\}]\\\beta\in\delta^{-1}[\{o'\}] }} d(\alpha,\beta)
\geq
\varepsilon_{o,o'}
\end{equation}
for all $o,o'\in D$ with $o\neq o'$.
Setting $\varepsilon = \inf_{o,o'\in D, o\neq o'} \varepsilon_{o,o'}$, we have $\varepsilon>0$ since~$D$ is finite.
Let~$T\geq0$ be the smallest nonnegative integer such that $2^{-T}<\varepsilon$.
For every open $2^{-T}$-ball $B_T(\alpha) \in V(\timeprotocol{T})$ there is a unique~$o\in D$ such that $B_T(\alpha)\subseteq \delta^{-1}[\{o\}]$:
Otherwise, if~$B_T(\alpha)$ contained vertices~$\beta$ and~$\gamma$ of at least two components $\delta^{-1}[\{o\}]$ and $\delta^{-1}[\{o'\}]$ with $o\neq o'$, respectively, then we would have $2^{-T}\geq \diam(B_T(\alpha))\geq d(\beta,\gamma) \geq \varepsilon_{o,o'} \geq \varepsilon$, which is incompatible with the choice of~$T$.
We can thus unambiguously define $\delta_T(B_T(\alpha)) = o$.
It remains to show that~$\delta_T$ is simplicial, chromatic, and that~$\delta_T\circ\Xi_T$ is carried by~$\Delta$.
Chromaticity of~$\delta_T$ follows from the fact that $d(\alpha,\beta)=2 > 2^{-T}$ whenever $\alpha\in\cview_p$ and $\beta\in\cview_q$ with $p\neq q$, so all $\beta\in B_T(\alpha)$ have the same color~$p$, which is the same as that of~$o$ since~$\delta$ is chromatic.
To prove that~$\delta_T$ is simplicial, let~$\sigma$ be a simplex of~$\timeprotocol{T}$.
By definition of the complex~$\timeprotocol{T}$, there exists a simplex~$\tau$ of~$\protocol$ such that every element of~$\sigma$ contains one vertex of~$\tau$.
By definition of~$\delta_T$, the set~$\delta_T[\sigma]$ is equal to~$\delta[\tau]$, which is a simplex of~$\outputs$ since~$\delta$ is simplicial.
Likewise, we have $(\delta_T\circ\Xi_T)(\sigma) = (\delta\circ\Xi)(\sigma)$ for every input simplex~$\sigma$ in~$\inputs$, which is a subcomplex of~$\Delta(\sigma)$ since~$\delta\circ\Xi$ is carried by~$\Delta$.
Hence~$\delta_T\circ\Xi_T$ is carried by~$\Delta$ as well.

To prove (3)$\Rightarrow$(4), we define the following decision protocol for process~$p$:
In its $T$\textsuperscript{th} local step, process~$p$ decides the value $\delta_T(B_T(\alpha))$ where~$\alpha$ is any local-view sequence $\alpha\in\cview_p$ in which in the local views up to index~$T$ coincide with its own local views up to now in the current execution.
By the definition of the metric~$d_p$ on~$\cview_p$, all these~$\alpha$ lead to the same open $2^{-T}$-ball $B_T(\alpha)$.
The decision value of~$p$ is thus unambiguous.
The fact that the task specification is satisfied by this protocol follows as usual from the hypotheses on~$\delta_T$.
\end{proof}

On the other hand, if~$V(\protocol)$ is not compact, then the equivalence in Theorem~\ref{thm:projective:limit:compact} need not hold, as is shown by the example in Section~\ref{sec:need-conitnuity}.

\section{Conclusion}

We put together combinatorial and point-set topological arguments to 
prove a generalized asynchronous computability theorem,
which applies also to non-compact computation models.
This relies on showing that in non-compact models, protocols solving 
tasks correspond to simplicial maps that need to be \emph{continuous}.
We show an application to the \emph{set agreement} task. 
We also show that the usual finite-time 
protocol complex, where protocols and simplicial maps are the same, suffices 
for \emph{all} compact models.

It would be interesting to find other computation models and tasks 
where our techniques, and the generalized ACT, in particular,
can be applied.
Another intriguing direction for future research is to characterize which 
computation models lead to non-compact topological objects.

\newpage

\bibliographystyle{plain}
\bibliography{bibl}

\end{document}